\documentclass{article}
\usepackage{geometry}
\usepackage{subfigure}
\usepackage{amsfonts}
\usepackage{epsfig}
\usepackage{amssymb,amsmath,bm}
\usepackage{amsthm}
\usepackage{mathrsfs}
\usepackage{fancyhdr}
\usepackage[auth-sc,affil-it]{authblk}
\usepackage{enumerate,color}

\newcommand{\bbR}{\mathbb R}

\newtheorem{theorem}{Theorem}[section]
\newtheorem{lem}{Lemma}[section]
\newtheorem{rem}{Remark}[section]
\newtheorem{prop}{Proposition}[section]

\newcounter{hypA}
\newenvironment{hypA}{\refstepcounter{hypA}\begin{itemize}
  \item[({\bf A\arabic{hypA}})]}{\end{itemize}}
\newcommand{\Exp}{\mathbb{E}}

\begin{document}

\title{Error Bounds and Normalizing Constants for Sequential Monte Carlo
in High Dimensions} 

\author[1]{Alexandros Beskos}
\author[2]{Dan Crisan}
\author[3]{Ajay Jasra}
\author[4]{Nick Whiteley}
\affil[1]{\footnotesize{Department of Statistical Science, University College London, London WC1E 6BT, UK.}\newline
E-Mail:\,\emph{\texttt{alex@stats.ucl.ac.uk}}}
\affil[2]{\footnotesize{Department of Mathematics, Imperial College London, London, SW7 2AZ, UK.}\newline
E-Mail:\,\emph{\texttt{d.crisan@ic.ac.uk}}}
\affil[3]{\footnotesize{Department of Statistics \& Applied Probability, National University of Singapore, Singapore, 117546, Sg.}\newline
E-Mail:\,\emph{\texttt{staja@nus.edu.sg}}}
\affil[4]{\footnotesize{Department of Mathematics, University of Bristol, Bristol, BS8 1TW, UK.}\newline
E-Mail:\, \ \emph{\texttt{nick.whiteley@bristol.ac.uk}}}
\date{}
\maketitle

\begin{abstract}
In a recent paper \cite{beskos},  the Sequential Monte Carlo (SMC)  sampler introduced in  \cite{delm:06,jarzynski,neal:01}
has been shown to be asymptotically stable in the dimension of the state space $d$ at a cost that is only polynomial in $d$, when $N$ the number of Monte Carlo samples, is fixed.  More precisely, it has been established that the effective sample size (ESS) of the ensuing (approximate) sample and the Monte Carlo
error of fixed dimensional marginals will converge as $d$ grows, with a computational cost of $\mathcal{O}(Nd^2)$. In the present work, further results on SMC methods in high dimensions are provided as $d\to\infty$ and with $N$ fixed. We deduce an explicit bound on the Monte-Carlo error for estimates derived using the  SMC
sampler and the exact asymptotic relative $\mathbb{L}_2$-error of the estimate of the normalizing constant.  We also establish marginal propagation of chaos properties of the algorithm. The accuracy in high-dimensions of some approximate SMC-based filtering schemes is also discussed.
\end{abstract}

\section{Introduction}

High-dimensional probability distributions are increasingly of interest in a wide variety of applications. In particular, one is concerned with the estimation of expectations with respect to such distributions. Due to the high-dimensional nature of the probability laws, such integrations cannot typically be carried out analytically; thus practitioners will often resort to Monte Carlo methods.

An important Monte Carlo methodology is Sequential Monte Carlo samplers (see \cite{delm:06,neal:01}). This is a technique designed to approximate a sequence of densities defined on a common state-space. The method works by simulating a collection of $N\geq 1$ weighted samples
(termed particles) in parallel. These particles are propagated forward in time via Markov chain Monte Carlo (MCMC), using importance sampling (IS) to correct, via the weights, for the discrepancy between  target distributions and proposals. Due to the weight degeneracy
problem (see e.g.~\cite{doucet}), resampling is adopted, sometimes performed when the ESS drops below some threshold.
Resampling generates samples with replacement from the current collection of particles using the importance weights, resetting  un-normalized weights to 1 for each sample.
The ESS is a number between 1 and $N$ and indicates, approximately, the number of useful samples. For SMC samplers one is typically interested in sampling
a single target density on $\mathbb{R}^d$, but due to some complexity, a collection of artificial densities are introduced, starting at some easy to sample distribution and creating a smooth path to the final target.

Recently (\cite{bengtsson,bickel,snyder}) it was shown that some IS methods will not stabilize, in an appropriate sense, as the  dimension of target densities in a particular class grows, unless $N$ grows exponentially fast with dimension $d$. In later work, \cite{beskos} have established that the SMC sampler technique can be stabilized at a cost that is only polynomial in $d$. It was shown in \cite{beskos} that ESS and the Monte Carlo
error of fixed dimensional marginals stabilize as $d$ grows, with a cost of $\mathcal{O}(Nd^2)$. This corresponds to introducing $d$ artificial densities between an initial distribution and the one of interest.
The case of fixed $d$ also has been analyzed recently \cite{whiteley}.

The objective of this article is to provide
a more complete understanding of SMC algorithms in high dimensions, complementing and building upon the results of \cite{beskos}. A variety of results are presented, addressing some generic theoretical properties of the algorithms and some issues which arise from specific classes of application. 

\subsection{Problems Addressed}

The first issue investigated is the increase in error of estimating fixed-dimensional marginals using SMC samplers relative to i.i.d.~sampling. 
Considering the case when one resamples at the very final time-step we show that the $\mathbb{L}_2$-error increases only by a factor of $\mathcal{O}(N^{-1})$ uniformly in  $d$. Resampling at the very final time-step is often of importance in real applications; see e.g.~\cite{ddj_abc}.

The second issue  we address is
the estimation of ratios of normalizing constants approximated using SMC samplers. This is critical in many disciplines, including Bayesian or classical statistics, physics and rare events. In particular, for Bayesian model comparison,
Bayes factors are associated to  statistical models in high-dimensional spaces, and these Bayes factors need to be estimated by numerical techniques such as SMC.
The normalizing constant in SMC methods has been well-studied: see \cite{cerou1,delmoral}.
Among the interesting results that have been proved in the literature is the unbiased property. However, to our knowledge, no results have been proved in the context of asymptotics in dimension $d$. In this article we provide an expression for fixed $N$ of the relative $\mathbb{L}_2$-error of the SMC estimate of a ratio of normalizing constants. The algorithm can include resampling, whereby the expression differs. The rate of convergence is $\mathcal{O}(N^{-1})$, when the computational cost is $\mathcal{O}(Nd^2)$.
The results also allow us compare between different sequences of densities used within the SMC method. 

The third issue  we investigate is asymptotic independence properties of the particles when one resamples: propagation of chaos - see \cite[Chapter 8]{delmoral}. This issue has practical implications which we discuss below. 
It is shown that, in between any two resampling times, 
any fixed dimensional marginal distribution 
of any fixed block of $1\leq q\leq N$ particles
among the $N$ particles
are asymptotically independent with the correct marginal. 
This result is established as $d$ grows with $N$ fixed, whilst the classical results require $N$ to grow.
As in \cite{beskos,whiteley}, this establishes that the ergodicity of the Markov kernels used in the algorithm can provide stability of the algorithm, even in high dimensions if the number of artificial densities is scaled appropriately with $d$. 

The final issue we address is the problem of filtering (see Section \ref{sec:filtering} for a description). Ultimately, we do not provide any analysis of a practical SMC algorithm which stabilizes as $d$ increases. However, it is shown that when one inserts an SMC sampler in-between the arrival of each data-point and updates the entire collection of states, then the algorithm stabilizes as $d$ grows at a cost which is $\mathcal{O}(n^2Nd^2)$, $n$ being the time parameter. This is  of limited practical use, as the computational storage costs increase with $n$. Motivated by the SMC sampler results, we consider some strategies
which could be attempted to stabilize high-dimensional SMC filtering algorithms. In particular, we address two strategies which insert an SMC sampler in-between the arrival of each data-point. The first, which only updates the state at the current time-point, fails to stabilize as the dimension grows, unless the particles increase exponentially in the dimension. The second, which uses a marginal SMC approach (see e.g.~\cite{poyiadjis}) also exhibits the same properties.
At present we are not aware of any online (i.e.~one which has a fixed computational cost per time step) SMC algorithm which can provably stabilize with $d$ \emph{for any model}, unless $N$ is exponential in $d$. It is remarked, as noted in \cite{bickel}, that there exist statistical models for which SMC methods can work quite well in high-dimensions. In relation to this, we then investigate the SMC simulation of a filter based upon Approximate Bayesian Computation (ABC)~\cite{jasra}. The ABC approximation induces bias which cannot be removed in practice. We show here that the simulation error stabilizes as the dimension grows, but we argue that the bias is likely to explode as $d$ grows. This discussion is also relevent for the popular ensemble Kalman filter (EnKF) employed in high dimensional filtering problems in physical sciences (e.g.~\cite{evensen}).

The paper is structured as follows:  
In Section \ref{sec:smc_samplers} we describe the SMC sampler algorithm together with  our mathematical assumptions.
In Section \ref{sec:main_res} our main results are given. 
In addition, we introduce a general annealing scheme, coupled with a consideration
of stability results for data-point tempering \cite{chopin1}; this latter study connects with our discussion in Section \ref{sec:filtering} on filtering. Section \ref{sec:simos} considers the practical implications of our main results with numerical simulations.
The filtering problem is addressed in Section \ref{sec:filtering}.
We conclude with a summary in  Section \ref{sec:disc}. Most of the proofs of our results are given in the Appendix.

\subsection{Notation}

Let $(E,\mathscr{E})$ be a measurable space and $\mathscr{P}(E)$
the set of probability measures on it.
For $\mu$ a $\sigma-$finite measure on $(E,\mathscr{E})$ and $f$ a measurable function,
we set $\mu(f)=\int_E f(x)\mu(dx)$.
For $\mu\in\mathscr{P}(E)$ and $P$ a Markov kernel on
$(E,\mathscr{E})$, we use the integration notation 
$P(f)(x)=\int_E P(x,dy)f(y)$ and
$\mu P(f) =$ $\ \int_{E} \mu(dx) P(f)(x)$. In addition, $P^n(f)(x) := \int_{E^{n-1}}P(x,dx_1)P(x_1,dx_2)\times\cdots\times P(f)(x_{n-1})$. The total variation difference norm for $\mu,\lambda\in\mathscr{P}(E)$
is $\|\mu-\lambda\|_{tv}:=\sup_{A\in\mathscr{E}}|\mu(A)-\lambda(A)|$.
The class of bounded (resp.~continuous and bounded) measurable functions $f:E\rightarrow\mathbb{R}$ is written $\mathcal{B}_b(E)$ (resp.~$\mathcal{C}_b(E)$).
For $f\in\mathcal{B}_b(E)$, we write $\|f\|_{\infty}:=\sup_{x\in\mathbb{R}}
|f(x)|$.
We will denote the $\mathbb{L}_{\varrho}$-norm of random variables as $\|X\|_{\varrho} = 
\Exp^{1/\varrho}\,|X|^{\varrho}$ with 
$\varrho\ge 1$.  For a given vector $(x_1,\dots,x_p)$ and $1\leq q\leq s\leq p$ we denote by $x_{q:s}$
the sub-vector $(x_q,\dots,x_s)$. For a measure $\mu$ the $N$-fold product is written $\mu^{\otimes N}$. For any collection of functions $(f_k)_{k\geq 1}$, $f_k:E\rightarrow\mathbb{R}$, we write 
$f_1\otimes\cdots\otimes f_k:E^{k}\rightarrow\mathbb{R}$ for their tensor product.
Throughout $M$ is used to denote a constant whose meaning may change, depending upon the context; important dependencies are written as $M(\cdot)$.
In addition, all of our results hold on probability space $(\Omega,\mathscr{F},\mathbb{P})$,
with $\mathbb{E}$ denoting the expectation operator and $\mathbb{V}\textrm{ar}$ the variance.
Finally, $(\Rightarrow)$ denotes convergence in distribution.

\section{Framework}\label{sec:smc_samplers}


\subsection{Algorithm and Set-Up}

We consider the scenario when one wishes to sample from a target distribution with density $\Pi$ on $E^d$ 
($E\subseteq \mathbb{R}$)
with respect to Lebesgue measure,
known point-wise up to a normalizing constant.
In order to sample from $\Pi$, we introduce a sequence of `bridging' densities which start from an easy to sample target and  evolve toward $\Pi$. In particular, we will consider the densities:
\begin{equation}
\label{eq:aux}
\Pi_n(x) \propto \Pi(x)^{\phi_n}\ , \quad x\in E^d\ ,
\end{equation}
for
$
0\leq \phi_0<\cdots < \phi_{n-1}< \phi_n < \cdots<\phi_p=1. 
$
Below, we use the short-hand $\Gamma_n$ to denote un-normalized densities associated to $\Pi_n$.

One can sample from $\{\Pi_n\}$ using an SMC sampler
that targets the sequence of densities:
\begin{equation*}
\widetilde{\Pi}_n(x_{1:n}) = \Pi_n(x_n) \prod_{j=1}^{n-1} L_j(x_{j+1},x_j)
\end{equation*}
with domain $(E^d)^n$ of dimension that increases with $n=1,\ldots,p$; 
here, $\{L_n\}$ is a sequence of artificial backward Markov kernels that can, in principle, be arbitrarily selected (\cite{delm:06}).
Let $\{K_n\}$ be a sequence of Markov kernels of invariant density $\{\Pi_n\}$ and $\Upsilon$ a distribution; assuming the weights appearing in the statement of the algorithm are well-defined Radon Nikodym derivatives, the SMC algorithm 
we will ultimately explore is the one defined in Figure \ref{tab:SMC}. It is remarked that our analysis is not necessarily constrained to the case of resampling according to ESS. 
\begin{figure}[h]
\begin{flushleft}
\medskip
\hrule
\medskip
{\itshape
\begin{enumerate}
\item[\textit{0.}] Sample $X_0^1,\dots X_0^N$ i.i.d.\@ from $\Upsilon$ and compute the weights for each particle $i\in\{1,\dots,N\}$:
\begin{equation*}
w_{0:0}^{i} = \tfrac{\Gamma_0(x_0^i)}{\Upsilon(x_0^i)}\ .
\end{equation*}
Set $n=1$ and $l=0$.
\vspace{0.1cm}
\item[\textit{1.}]  If $n\leq p$, for each $i$
sample $X_n^i\mid X_{n-1}^{i}=x_{n-1}^i$ from $K_n(x_{n-1}^i,\,\cdot)$ and calculate the weights: 
\begin{equation*}
w_{l:n}^i  =
\tfrac{\Gamma_n(x_{n-1}^i)}{\Gamma_{n-1}(x_{n-1}^i)}\,w^{i}_{l:(n-1)}\ .
\end{equation*}
Calculate the Effective Sample Size (ESS):
\begin{equation}
\label{eq:ess_def}
\textrm{ESS}_{\,l:n}(N) = \tfrac{\left(\sum_{i=1}^N w_{l:n}^i\right)^2}{\sum_{i=1}^N (w_{l:n}^i)^2} \ . 
\end{equation}
If $\textrm{ESS}_{\,l:n}(N)<a$: \\ 
\hspace{0.3cm} resample particles according to their normalised 
weights 
\begin{equation}
\label{eq:normed}
\overline{w}_{l:n}^i = \tfrac{w_{l:n}^i}{\sum_{j=1}^{N}w_{l:n}^j}\ ; 
\end{equation}
\hspace{0.3cm} set $l=n$ and re-initialise the weights by setting $w_{l:n}^i\equiv 1$, $1\le i \le N$;\\
\hspace{0.3cm} let $\check{x}_{n}^{1},\ldots,\check{x}_n^{N}$ now denote the resampled particles.\\
Set $n=n+1$. \\
Return to the start of Step 1.
%
%
\end{enumerate} }
\medskip
\hrule
\medskip
\end{flushleft}
\vspace{-0.4cm}
\caption{The SMC samplers algorithm analyzed in this article.}
\label{tab:SMC}
\end{figure}

For simplicity, we will henceforth  assume  that $\Upsilon\equiv \Pi_0$. It should be noted that when $\Upsilon$ is different from
$\Pi_0$, one can modify the sequence of densities to  a bridging scheme which moves from $\Upsilon$ to~$\Pi$. However, in practice, one can make $\Pi_0$ as simple as possible so we do not consider this possibility; see
\cite{whiteley} for more discussion and analysis when $d$ is fixed (note that our results for SMC samplers will hold,
 with some modifications, also  in this scenario).
Note, that we only consider here the multinomial resampling method. 
%

We will investigate the stability of SMC estimates associated to the algorithm in Figure \ref{tab:SMC}.
To obtain analytical results we will need to simplify the structure of the algorithm. In particular, we will consider 
an i.i.d.~target: 
\begin{equation}
 \label{eq:target}
\Pi(x) = \prod_{j=1}^{d}\pi(x_j)\ ; \quad \pi(x_j) =  \exp\{g(x_j)\}\ , 
\end{equation}
with  $x_j\in E$,
for some $g:E\mapsto \bbR$. In such a case all bridging densities are also i.i.d.: 
\begin{equation*}
\Pi_{n}(x) \propto \prod_{j=1}^{d}\pi_{n}(x_j)\ ;\quad \pi_n(x_j) \propto \exp\{\phi_n\,g(x_j)\} \ . 
\end{equation*}
It is remarked that this assumption is made for mathematical convenience: see \cite{beskos} for a discussion on this. 
A further assumption that will facilitate the mathematical analysis is to apply independent kernels along the different co-ordinates. That is, we will assume:     
\begin{equation*}
K_{n}(x,dx') =  \prod_{j=1}^d k_{n}(x_j,dx_j') \ ,
\end{equation*}
where each transition kernel $k_n(\cdot, \cdot)$ preserves 
$\pi_n(x)$;
that is, $\pi_{n}k_{n}=\pi_{n}$.
We study the case when one selects cooling constants $\phi_n = \phi_{n}(d)$ and $p=p(d)$ as below: 
\begin{equation}
\label{eq:tune}
p = d\  ; \quad \phi_{n} (=\phi_{n,d})= \phi_0 + \tfrac{n(1-\phi_0)}{d}\ ,\quad  0\le n\le  d \ ,
\end{equation} 
with $0\le \phi_0<1$ given and  fixed with respect to $d$.  It is possible, with only notational changes, to consider (as in \cite{whiteley}) the case when the annealing sequence is derived via a more general non-decreasing Lipschitz function; see Section \ref{sec:annealing_seq}.
As in \cite{beskos}, it will be convenient to 
consider the continuum of invariant densities and kernels on the whole of the time interval 
$[\phi_0,1]$. So, we will set: 
\begin{equation*}
\pi_s(x) \propto \pi(x)^{s} = \exp\{s\, g(x)\}\ , \quad s \in [\phi_0,1] \ .
\end{equation*}
Similarly $k_s(x,dx')$ with $s\in (\phi_0,1]$ is the continuous-time version of the kernels $k_n(x,dx')$.
As in \cite{beskos}, the mapping $l_d(s)= \lfloor \tfrac{d\,(s-\phi_0)}{1-\phi_0}\rfloor$ is used to move between continuous and discrete time.

\subsection{Conditions}
\label{sect:conditions}
We state the conditions under which we will derive our results.
We will require that $E\subset \bbR$ with $E$ being \emph{compact}.
The conditions below correspond to a simplification of the weaker conditions in \cite{beskos}
under the scenario of the compact state space $E$ that we consider here. 
We note that imposing compactness has been done mainly to simplify proofs and keep them at a reasonable length. The numerical examples later on are executed on unbounded state spaces, and 
do not seem to invalidate our conjecture that several of the results in the sequel will also hold on 
 unbounded spaces under appropriate geometric ergodicity conditions, as it was the case for the stability results 
as $d\rightarrow \infty$ in 
\cite{beskos}.
We remark that all results of \cite{beskos}
also hold under the assumptions stated here.

\begin{hypA} 
\label{hyp:A}
{\em Stability of $\{k_{s}\}$ - Uniform Ergodicity.}\vspace{0.3cm} \\ 
There exists a constant $\theta\in(0,1)$ and some $\varsigma\in \mathscr{P}(E)$ such that for each 
$s\in(\phi_0,1]$ the state-space $E$ is $(1,\theta,\varsigma)$-small with respect to $k_s$.
\end{hypA}

\begin{hypA} 
\label{hyp:B}
{\em Perturbations of $\{k_{s}\}$.} \vspace{0.3cm} \\
There exists an $M<\infty$ such that 
for any $s,t\in(\phi_0,1]$ we have
$
\|k_s-k_t\|_{tv}  \leq M\,|s-t|\ .
$
\end{hypA}
Note that the statement that $E$ is $(1,\theta,\varsigma)$-small w.r.t.\@ to $k_s$ means that $E$ is a one-step small set for the Markov kernel, with minorizing distribution $\varsigma\in\mathscr{P}(E)$ and parameter $\theta\in(0,1)$ (i.e.~$k_s(x,A)\geq \theta\,\varsigma(A)$ for each $(x,A)\in E\times\mathscr{E}$). 

In the context of our analysis, we will consider an SMC algorithm that resamples at the \emph{deterministic} times $t_1(d),\dots,t_{m^*(d)}(d)\in[\phi_0,1]$ (i.e.\@ resamples after $n=l_d(t_{k}(d))$ steps for $k=1,2,\ldots, m^*(d)$)
such that $t_0(d)=\phi_0$ and $t_0(d)<t_1(d)<\cdots<t_{m^*(d)}(d)< t_{m^*(d)+1}(d)=1$, with $l_d(t_{m^*}(d))<d$.
We will also assume that as $d\rightarrow \infty$ we have that $m^*(d)\rightarrow m^*$ 
and $t_k(d)\rightarrow t_k$ for $t_k\in[\phi_0,1]$ for all relevant~$k$.
Such deterministic times are meant to mimic the behaviour of randomised ones
(i.e.~as for the case of the original algorithm in Figure \ref{tab:SMC}) and provide 
a mathematically convenient framework for understanding the impact of resampling on the properties of 
the algorithm. 
Examples of such times can be found in \cite{beskos,delmoral_resampling}; 
the results therein provide an approach for converting the results for deterministic times, to randomized ones.
In particular, they show that with a probability converging to 1 as $N\rightarrow\infty$ the randomized times
essentially  coincide with the deterministic ones.
We do not consider that here as it would follow a similar proof to \cite{beskos}, depending upon how the resampling times are defined.
(An alternative procedure of treating dynamic resampling times is to use the construction in \cite{arnaud}; this is not considered here.)
For simplicity, we will henceforth assume that $d$ is large enough so that $m^{*}(d)\equiv m^{*}$.

\subsection{Log-Weight-Asymptotics}\label{sec:log_weight}

Given the set-up \eqref{eq:tune}
and the resampling procedure at the deterministic times $t_1(d),\dots,t_{m^{*}}(d)\in[\phi_0,1]$, and due to the i.i.d.\@ structure described above, we have the following expression for the particle weights:
\begin{equation*}
\log(w_{l_d(t_{k-1}(d)):l_d(t_{k(d)})}^i) =  \tfrac{1}{d}\sum_{j=1}^{d}\bar{G}_{k,j}^i  \ ,
\end{equation*}
where
$\bar{G}_{k,j}^i = (1-\phi_0)\sum_{n=l_d(t_{k-1}(d))}^{l_d(t_k(d))-1}g(X_{n,j}^i)$
for $1\le i \le N$. 
The work in \cite{beskos} illustrates stability of the normalised weights as $d\rightarrow\infty$. Define 
the standardised log-weights:
\begin{equation}
\label{eq:defineGG}
G_{k,j}^i = (1-\phi_0)\sum_{n=l_{d}(t_{k-1}(d))}^{l_{d}(t_{k}(d))-1}\,\big(\,g(X_{n,j}^i)-
\mathbb{E}_{ \pi_{t_{k-1}(d)}}[\,g(X_{n,j}^i)\,]\,\big)\ .
\end{equation}
The notation 
$\Exp_{\pi_{t_{k-1}(d)}}[\,g(X_{n,j}^{i})\,]$ refers to an expectation under the initial dynamics
$X_{l_d(t_{k-1}(d)),j}^i\sim \pi_{t_{k-1}(d)}$; after that,
$X_{n,j}^{i}$ will evolve according to the Markov transitions $k_{n}$. 
We also use the notation $\Exp_{\pi^{\otimes Nd}_{t_{k-1}(d)}}[\,\cdot\,]$ when imposing similar initial dynamics, but now independently over all co-ordinates and particles; such dynamics differ of course from the 
actual particle dynamics of the SMC algorithm.
In what follows, we use the Poisson equation:
\begin{equation*}
g(x) -  \pi_u(g) =  \widehat{g}_u(x) - k_u(\widehat{g}_u)(x) \ 
\end{equation*}
and in particular the variances:
\begin{equation}
\label{eq:varia}
\sigma^2_{s:t}= 
(1-\phi_0)\int_{s}^{t}
\pi_u(\widehat{g}_u^2-k_u(\widehat{g}_u)^2)du\ , \quad \phi_0\leq s<t\leq 1\ .
\end{equation}
The following weak limit can be derived from the proof of Theorem 4.1 of \cite{beskos}.

\begin{rem}[Log-Weight-Asymptotics]
\label{th:CLT1}

Assume (A\ref{hyp:A}-\ref{hyp:B}) and $g\in\mathcal{B}_b(E)$. For any $N\geq 1$  we have:
\begin{equation*}
\Big( \tfrac{1}{d}\,\sum_{j=1}^{d} G_{k,j}^{i}\, \Big)_{i=1}^{N} \Rightarrow (\,Z^i\,)_{i=1}^{N}\ ,
\end{equation*}
where the $Z^i$'s are  i.i.d.\@ copies from $N(0,\sigma^{2}_{t_{k-1}:t_k})$.
\end{rem}
The result illustrates that the consideration of $\mathcal{O}(d)$ Markov chain steps between resampling times 
stabilise the particle standardised log-weights as $d\rightarrow \infty$.

\section{Main Results}
\label{sec:main_res}

We now present the main results of the article.

\subsection{Asymptotic Results as $d\rightarrow \infty$}

\subsubsection{$\mathbb{L}_2-$Error}

The first result of the paper pertains to the Monte-Carlo error from estimates derived via the SMC method.
We will consider mean squared errors and obtain $\mathbb{L}_2$-bounds with resampling carried out also `at the end', that is when one resamples also at time $t=1$.
Below, recall that the $\check{X}$-notation is for resampled particles.
Resampling at time $t=1$ is required when one wishes to obtain un-weighted samples. We have the following result, with proof in Appendix \ref{app:l2proof}.

\begin{theorem}\label{theo:l2resampling}
Assume (A\ref{hyp:A}-\ref{hyp:B}) and $g\in\mathcal{B}_b(E)$. Then for any $N\geq 1$, $\varphi\in\mathcal{C}_b(E)$ we have:
\begin{equation*}
\Big\|\,\Big(\,\tfrac{1}{N}\sum_{i=1}^N\big[\,\varphi(\check{X}_{d,1}^i)-\pi(\varphi)\,\big]\,\Big)\,\Big\|^2_2 \leq 
\mathbb{V}\textrm{\emph{ar}}_{\pi}[\varphi]\,\tfrac{1}{N}\,\big(\,
1 + M(\sigma^2_{t_{m^*-1}:1})\,
\big)
\end{equation*}
for
\begin{equation*}
M(\sigma^2_{t_{m^*-1}:1}) = e^{ \sigma^2_{t_{m^*-1}:1}}
+M\,e^{17\sigma^2_{t_{m^*-1}:1}}\,\tfrac{1}{N^{1/6}} 
\end{equation*}
with $M<\infty$ independent of $N$ and $\sigma^2_{t_{m^*-1}:1}$.
\end{theorem}

\begin{rem} Compared to the  i.i.d.~sampling scenario, the upper bound contains the additional term  
$\textrm{\emph{Var}}_{\pi}[\varphi]\,\tfrac{1}{N}\,M(\sigma^2_{t_{m^*-1}:1})$.
 This is a bound on the cost induced due to the dependence of the particles.
\end{rem}

\subsubsection{Normalizing Constants}

The second main result of the paper is the stability of estimating normalising constants in high dimensions.
The quantity of interest here is the ratio of normalising constants:
\begin{equation}
\label{eq:norm}
c_d:=\frac{\int_{E^{d}}\Pi(x)dx}{\int_{E^{d}}\Pi_{\phi_0}(x)dx}\ .
\end{equation}
We first consider the SMC sampler in Figure \ref{tab:SMC} with no resampling. Define:
\begin{equation*}
\gamma^N_d(1) = \tfrac{1}{N}\sum_{i=1}^{N}w^{i}_{0:d}\equiv  \tfrac{1}{N}\,\sum_{i=1}^N e^{\frac{1}{d}\sum_{j=1}^d \bar{G}_j^i} \ ; \quad 
\gamma_d(1)= \mathbb{E}\,\big[\,e^{\frac{1}{d}\sum_{j=1}^d \bar{G}_j^1}\,\big]\ ,
\end{equation*}
where
$\bar{G}_j^i = (1-\phi_0)\sum_{n=0}^{d-1}g(X_{n,j}^i)$. 
From standard properties of SMC we have $\Exp\,[\,\gamma^N_d(1)\,]=\gamma_d(1)\equiv c_d$. 
Now, consider the relative $\mathbb{L}_2$-error:
%
\begin{equation*}
\mathbb{V}_2(\gamma_d(1))=\mathbb{E}\,\big[\,\big(\tfrac{\gamma^N_d(1)}{\gamma_d(1)}-1\,\big)^2\,\big]\ .
\end{equation*}
We then have the following result, proven in Appendix \ref{app:nc}. 

\begin{theorem}\label{theo:nc}
Assume (A1-2) and $g\in\mathcal{B}_b(E)$.
Then for any $N\geq 1$:
$$
\lim_{d\rightarrow\infty}\mathbb{V}_2(\gamma_d(1)) = \tfrac{e^{\sigma_{\phi_0:1}^2}-1}{N}\ .
$$
\end{theorem}
%
\noindent The result establishes a $\mathcal{O}(N^{-1})$ rate of convergence at a computational cost of $\mathcal{O}(Nd^2)$. The information in the limit is in terms of the expression $\sigma_{\phi_0:1}^2$. As in \cite{beskos},
this is a critical quantity, which helps to measure the rate of convergence of the algorithm.

We now consider the SMC sampler in Figure \ref{tab:SMC} with resampling at the deterministic times 
$\{t_k(d)\}$ described in Section \ref{sect:conditions}. We make the following definitions:
\begin{equation*}
\gamma_{d,k}^N(1) = \tfrac{1}{N}\sum_{i=1}^Ne^{\frac{1}{d}
\sum_{j=1}^d \bar{G}_{k,j}^i}\ ; \quad \gamma_{d,k}(1) = \mathbb{E}_{\pi_{t_{k-1}(d)}^{\otimes Nd}}\,\big[\,e^{\frac{1}{d}
\sum_{j=1}^d \bar{G}_{k,j}^1}\,\big],
\end{equation*}
where $k\in\{1,\dots,m^*+1\}$ and
$
\bar{G}_{k,j}^i$ as defined in Section \ref{sec:log_weight}. As in the non-resampling case, 
we again have the unbiasedness property for the estimate of $c_d$:
$\Exp\,\big[\,\prod_{k=1}^{m^{*}+1}\gamma_{d,k}^N(1)\,\big] = \prod_{k=1}^{m^{*}+1}\gamma_{d,k}(1) \equiv c_d\ .$
We have the following result whose proof is in Appendix \ref{app:nc}. 

\begin{theorem}
\label{theo:nc1}
Assume (A\ref{hyp:A}-\ref{hyp:B}) and $g\in\mathcal{B}_b(E)$.
Then for any $N\geq 1$:
$$
\lim_{d\rightarrow\infty} \mathbb{V}_2\Big(\prod_{k=1}^{m^*+1}\gamma_{d,k}(1)\Big) = 
e^{-\sigma_{\phi_0:1}^2} \prod_{k=1}^{m^*+1} \big[\,\tfrac{1}{N}e^{2\sigma^2_{t_{k-1}:t_k}}+\left(1-\tfrac{1}{N}\right)e^{\sigma^2_{t_{k-1}:t_k}}\,\big] - 1.
$$
\end{theorem}

\begin{rem}
In comparison with the no resampling scenario, the limiting expression here depends upon the incremental variance expressions.  On writing the limit in the form:
$$
e^{-\sigma_{\phi_0:1}^2}\prod_{k=1}^{m^*+1} 
e^{\sigma^2_{t_{k-1}:t_k}}
\big[1 + \tfrac{1}{N}\{e^{\sigma^2_{t_{k-1}:t_k}}-1\}\,\big] - 1
$$
if $N>(m^*+1)(e^{\overline{\sigma}^2}-1)$, with $\overline{\sigma}^2=\max_k\sigma^2_{t_{k-1}:t_k}$, then using similar manipulations to \cite[Corollary 5.2]{cerou1}, we have that
$$
\lim_{d\rightarrow\infty} \mathbb{V}_2\Big(\prod_{k=1}^{m^*+1}\gamma_{d,k}(1)\Big) \leq  \frac{2(m^*+1)(e^{\overline{\sigma}^2}-1)}{N}
$$
hence, the no resampling scenario has a lower error if this is less than the limit in Theorem \ref{theo:nc}. The upper-bound also shows that the error seems to grow with number of times one resamples.
The limit in Theorem \ref{theo:nc1} corresponds to the behavior of $\mathbb{E}_{\pi_{t_{k-1}(d)}^{\otimes Nd}}[\,\gamma_{d,k}^N(1)^2/\gamma_{d,k}(1)^2\,]$ between each resampling time. In effect, the ergodicity of the system takes over, and breaks up the error in estimation of the ratio of normalizing constants to different tours between resampling times (see Proposition \ref{theo:asymp_indep}).
\end{rem}
\noindent To provide some intuition for Theorem \ref{theo:nc1}, we give the following result.
It is associated to the asymptotic independence, between resampling times,
of the log-weights
$\sum_{j=1}^{d}\frac{1}{d}\,G_{k,j}^i$ in (\ref{eq:defineGG}).
The proof of the following can be found in Appendix \ref{app:nc}.

\begin{prop}
\label{theo:asymp_indep}
Assume (A\ref{hyp:A}-\ref{hyp:B}) and that $g\in\mathcal{B}_b(E)$. Then 
for any $N\geq 1$, $i\in\{1,\dots,N\}$,  $c_{1:k}\in\mathbb{R}$ and $k\in\{1,\dots,m^*+1\}$, we 
have that:
$$
\lim_{d\rightarrow\infty} \mathbb{E}\,\big[\,e^{\sum_{l=1}^k c_l\,\frac{1}{d}\,\sum_{j=1}^d G_{l,j}^i}\,\big]=
\prod_{l=1}^k \Exp\,[\,e^{c_l Z^l}\,]\ ,
$$
where $Z^{l}\sim N(0,\sigma^2_{t_{l-1}:t_l})$ independently over $1\le l \le k$.
\end{prop}

\subsubsection{Propagation of Chaos}
Finally we deduce a rather classical result in the analysis of particle systems: propagation of chaos, demonstrating the asymptotic (in the classical setting) independence of any fixed block of $q$ of $N$ particles as $N$ grows.
The following scenario, with the SMC sampler with resampling (at the times $\{t_k(d)\}$), is considered: let $s(d)$ be a sequence
 such that $s(d)\in(t_{k-1}(d),t_k(d))$ for some $1\leq k \leq m^*+1$, with limit $s\in(t_{k-1}(d),t_k(d))$. Denote by $\mathbb{P}_{s(d),j}^{(q,N)}$ the marginal law of any of the $q$
particles out of $N$ at time $s(d)$ and in dimension $j\in\{1,\dots,d\}$.
 By construction, particles are considered at a time when they are not resampled. 
We have the following Propagation of Chaos result, whose proof is in Appendix \ref{app:A}.

\begin{prop}\label{theo:prop_chaos}
Assume (A\ref{hyp:A}-\ref{hyp:B}) and that $g\in\mathcal{B}_b(E)$. Then, for any fixed $j\in\{1,\dots,d\}$ and  
any $1\leq k \leq m^*+1$, a sequence $s(d)\in(t_{k-1}(d),t_k(d))$ with $s(d)\rightarrow s$,
and any $N\geq 1$, with
$1\leq q \leq N$ fixed:
\begin{equation*}
\lim_{d\rightarrow\infty}\|\,\mathbb{P}_{s(d),j}^{(q,N)}-\pi_{s}^{\otimes q}\,\|_{tv} = 0\ .
\end{equation*}
\end{prop}
The result establishes the asymptotic independence of the marginals of the particles, between any two resampling times, as $d$ grows. This is in contrast to the standard scenario where the particles only become independent as $N$ grows. Critically, the MCMC steps provide the effect that the marginal particle distributions converge 
to the target $\pi_s^{\otimes q}$. Thus,
Proposition \ref{theo:prop_chaos} establishes that it is essentially the ergodicity of the system which helps to drive the stability properties of the algorithm.
It should be noted that if one considers the particles just after resampling, one cannot obtain an asymptotic independence in $d$. Here, as in classical results for particles methods, one has to rely on increasing $N$.

\subsection{Other Sequences of Densities}

In this section, we discuss some issues associated with the selection of the sequence of chosen bridging  densities $\{\Pi_n\}$.

\subsubsection{Annealing Sequence}
\label{sec:annealing_seq}

Recall that we use the equidistant annealing sequence $\phi_n$  in \eqref{eq:tune}. However, one could also
consider a general differentiable, increasing Lipschitz function $\phi(s)$, $s\in[0,1]$ with $\phi(0)=\phi_0\ge 0$,  $\phi(1)=1$, and use the construction
$\phi_{n,d}=\phi(n/d)$; 
then the asymptotic result in Theorem \ref{th:CLT1} (and others that will follow) generalised to the choice 
of $\phi_{n,d}$ considered here would involve the variances:
\begin{equation}
\sigma_{s:t}^{2,\phi}= \int_{s}^{t}
\pi_{\phi(u)}(\widehat{g}_{\phi(u)}^2-k_{\phi(u)}(\widehat{g}_{\phi(u)})^2)
\bigg[\frac{d\phi(u)}{du}\bigg]
d\phi(u)\ , \quad 0\leq s < t \leq 1\label{eq:general_variance}\ ,
\end{equation}
in the place of $\sigma^2_{s:t}$ in (\ref{eq:varia}).
Our proofs in this paper are given in terms of the annealing sequence \eqref{eq:tune}, corresponding to a linear choice of $\phi(\cdot)$, but it is straightforward to modify them to the above scenario.

This point is illuminated by our main results. For example Theorem \ref{theo:nc} helps to compare various annealing schemes for estimating normalizing constants, via the limiting quantity \eqref{eq:general_variance}. That is, if we are only concerned with variance one prefers an annealing scheme which discretizes $\phi(s)$ versus one which discretizes $\nu(s)$ 
(differentiable monotonic increasing Lipschitz function with $\nu(0)=\nu_0>0$, $\nu(1)=1$)
for estimating normalizing constants if
$$
\tfrac{e^{\sigma_{0:1}^{2,\phi}}-1}{N} \leq 
\tfrac{e^{\sigma_{0:1}^{2,\nu}}-1}{N} \quad \Longleftrightarrow \quad
\sigma_{0:1}^{2,\phi} \leq \sigma_{0:1}^{2,\nu}\ .
$$
In practice, however, one has to numerically approximate $\sigma_{0:1}^{2,\phi}$ and $\sigma_{0:1}^{2,\nu}$, lessening the practical impact of this result. Similar to the scenario of Theorem \ref{theo:nc}, one can use Theorem \ref{theo:nc1}
to compare annealing schemes $\phi(s)$ and $\nu(s)$ 
(which potentially generate a different collection and number of limiting times $0<t_{1}^{\phi}<\cdots<t_{m_{\phi}^*}^{\phi}\leq 1$, $m_{\phi}^*$ and $0<t_{1}^\nu<\cdots<t_{m_{\nu}^*}^\nu\leq 1$,
$m_{\nu}^*$ respectively)
via the inequality
$$
e^{-\sigma_{0:1}^{2,\phi}} \prod_{k=1}^{m_{\phi}^*+1} \Big[\,\tfrac{1}{N}e^{2\sigma^{2,\phi}_{t_{k-1}^{\phi}:t_k^{\phi}}}
+\tfrac{N-1}{N}\,e^{\sigma^{2,\phi}_{t_{k-1}^{\phi}:t_k^{\phi}}}\Big]
\leq 
e^{-\sigma_{0:1}^{2,\nu}} \prod_{k=1}^{m_{\nu}^*+1} \Big[\,\tfrac{1}{N}e^{2\sigma^{2,\nu}_{t_{k-1}^{\nu}:t_k^{\nu}}}
+\tfrac{N-1}{N}e^{\sigma^{2,\nu}_{t_{k-1}^{\nu}:t_k^{\nu}}}\,\Big]
$$
but again, both quantities are difficult to calculate.

\subsubsection{Data Point Tempering}\label{sec:datapointtemp}

An interesting sequence of densities introduced in \cite{chopin1} arises in the scenario when $\Pi$ is 
associated  with a batch data-set
$y_1,\dots,y_p$. The idea is to construct the sequence of densities so that arriving data-points are added sequentially to the target as the time parameter of the algorithm increases. More concretely, 
we will assume here that:
\begin{equation*}
\Pi(x) \propto \exp\bigg\{\sum_{k=1}^p \sum_{j=1}^d g(y_k,x_j)\bigg\}\,,\quad x_j\in E\ ,
\end{equation*}
that is a density that is i.i.d.~in both the data and dimension.
In this scenario 
one could then adopt a sequence of densities of the form:
\begin{equation*}
\Pi_n(x) \propto \exp\bigg\{\sum_{k=1}^n \sum_{j=1}^d g(y_k,x_j)\bigg\}\ , \quad 1\leq n \leq p\ .
\end{equation*}
As noted also in Remark 3.5~of \cite{beskos}, for $d\rightarrow\infty$  
one cannot stabilize the associated SMC algorithm (described in Figure~\ref{tab:SMC}) as the ratio $\Pi_{n+1}/\Pi_n$ explodes for increasing $d$. 
To stabilize the algorithm as $d$ grows one can insert $\lfloor d/p\rfloor$ annealing steps between consecutive data points, thus forming the densities:
\begin{equation*}
\Pi^{(n)}_{k}(x) \propto \exp\bigg\{\phi\big(\,\tfrac{k}{\lfloor d/p\rfloor}\,\big)
\sum_{j=1}^d g(y_n,x_j)\bigg\} 
\times 
\Pi_{n-1}(x)
\ , \quad n\in\{1,\dots,p\}\,,\,\,\, k\in\{1,\dots,\lfloor \tfrac{d}{p}\rfloor\}\ ,
\end{equation*}
where $\phi(s)$ is as in Section \ref{sec:annealing_seq} with $\phi_0=0$. 
Then one can adopt
Markov kernels $K^{(n)}_{s}$, $1\leq n \leq p$, of product form with each component kernel 
$k^{(n)}_{s}$
having invariant measure  
\begin{equation*}
\pi_{s}^{(n)}(x) \propto \exp\bigg\{s\cdot g(y_n,x) + \sum_{k=1}^{n-1}g(y_k,x) \bigg\}\ , \quad x\in E\ .
\end{equation*}
We consider the scenario where there is no resampling and
denote by $\textrm{ESS}_{(0,d)}(n,N)$ the effective sample size with $n$ data after 
$\lfloor d/p\rfloor     $ steps of the $n$-th SMC sampler.
Throughout the data are taken as fixed.
We have the following result, which follows directly from Theorem 3.1.~of \cite{beskos}
and illustrates the stability of $\textrm{ESS}_{(0,d)}(n,N)$ as $d\rightarrow\infty$.

\begin{prop} 
Assume conditions (A\ref{hyp:A}-\ref{hyp:B}) for the kernels $k_s^{(n)}$, $n\in\{1,\dots,p\}$. Suppose
that for each data index $n\in\{1,\dots,p\}$, $g(y_n,\cdot)\in\mathcal{B}_b(E)$.
Then, for any fixed $N>1$ and $n\geq 1$, $\textrm{\emph{ESS}}_{(0,d)}(n,N)$ converges in distribution to: 
\begin{equation*}
\frac{\big(\,\sum_{i=1}^N e^{Z^{i}}\,\big)^2}{\sum_{i=1}^N e^{2 Z^i}}
\end{equation*}
where $Z^{i} \stackrel{i.i.d.}{\sim}N(0,\sigma^2)$ with
$$
\sigma^2 = \sum_{k=1}^n \int_{0}^1
\pi_{\phi(u)}^{(k)}((\widehat{g}_{\phi(u)}^{(k)})^2-k_{\phi(u)}^{(k)}(\widehat{g}_{\phi(u)}^{(k)})^2)
\bigg[\frac{d\phi(u)}{du}\bigg]
d\phi(u)\ .
$$
In particular,
\begin{equation*}
\lim_{d\rightarrow\infty}
\mathbb{E}\,\big[\,\textrm{\emph{ESS}}_{(0,d)}(n,N)\,\big]
=\mathbb{E}\bigg[\frac{\big(\,\sum_{i=1}^N e^{Z^{i}}\,\big)^2}{\sum_{i=1}^N e^{2 Z^{i}}}\bigg]\ .
\end{equation*}
\end{prop}

\noindent As for the case with annealing densities, there is a direct extension to the case where one resamples (see \cite{beskos}).
In addition, one can easily extend the results in Section \ref{sec:main_res} 
in the  data-point tempering case examined here. 
In connection to the filtering scenario, this is a class of densities that falls into the scenario of a state-space model with a deterministic dynamic on the hidden state (i.e.~only the initial state is stochastic and propagated deterministically; see e.g.~\cite{campillo}). This hints at algorithms for filtering which may stabilize as the dimension grows. As we shall see in Section \ref{sec:filtering}, it is not straight-forward to do this.

\section{Numerical Simulations}
\label{sec:simos}

We now present two numerical examples, to illustrate the practical implications of our theoretical results. It is noted that the state-space $E$ is not compact here, yet the impact of our results can still be observed. 

\subsection{Comparison of Annealing Schemes}\label{sec:nn}

We consider a target distribution comprised of $d$ i.i.d.\@ $N(0,1)$ co-ordinates.
The bridging densities are in this case:
\begin{equation}
\label{eq:avb}
\pi_{\phi(s)}(x) \propto \exp\big\{-\tfrac{1}{2}\phi(s)x^2\big\}\ .
\end{equation}
We will in fact consider two annealing schemes:
\begin{align}
\phi(s) & =  \phi_0 + (1-\phi_0)s\ ; \nonumber \\
\nu(s) & =  \tfrac{\phi_0 e^{\vartheta} - 1}{e^{\vartheta}-1} + \big(\tfrac{1-\phi_0}{e^{\vartheta} - 1}\big)e^{\vartheta s}\ . \label{eq:anna}
\end{align}
These are graphically displayed in Figure \ref{fig:asymp_prop} (a), with $\vartheta=5$. 

The purpose of investigating the two annealing schemes is as follows. In practical applications of SMC samplers we have observed that algorithms with slow initial annealing schemes can often out-perform those with faster ones (see Figure \ref{fig:asymp_prop} (a)). Thus, we expect scheme $\nu(s)$ to perform better than $\phi(s)$ w.r.t.~the expression for the asymptotic variance  \eqref{eq:general_variance}, hence deliver 
a lower relative $\mathbb{L}_2$-error for the estimation of the normalizing constant in high dimensions.
To obtain some analytically computable proxies for the asymptotic variances \eqref{eq:general_variance} we use the variances that one would obtain when $k_s(x,dx')\equiv \pi_s(dx')$, that is we substitute
$
\pi_{\phi(u)}(g^2) - \pi_{\phi(u)}(g)^2$ for  $\pi_{\phi(u)}(\widehat{g}_{\phi(u)}^2-k_{\phi(u)}(\widehat{g}_{\phi(u)})^2)
$
in \eqref{eq:general_variance}. In this scenario it is simple to show that, under the choice (\ref{eq:avb}),
we have that:
$$
\sigma^{2,\phi}_{0:1} = \frac{1}{2}\int_{0}^1 \bigg[\frac{1}{\phi(u)}\frac{d\phi(u)}{du}\bigg]^2 du\ .
$$
Figure \ref{fig:asymp_prop} (b) now plots the analytically available variances $\sigma^{2,\phi}_{0:1}$ and $\sigma^{2,\nu}_{0:1}$ (broken line) against $\phi_0$. The graph indeed provides some evidence that 
that the scheme $\nu(s)$ should give better results. This is particularly evident when $\phi_0$ is small; this is unsurprising as one initializes from $\Pi_{\phi_0}$, hence if $\phi_0$ is closer to 1  one expects a constant increase in the annealing parameter to be preferable to a slow initial evolution.

We ran SMC samplers with both annealing schemes with $N=10^4$ particles
and different dimension values $d\in\{10,25,50\}$.
The choice $\phi_0=\frac{1}{d}$ is used for both annealing schemes.
We used a Markov kernel $k_s(x,dx')$ corresponding to a Random-Walk Metropolis with proposal 
$y = x + N(0,\frac{1}{25\phi_0})$, thus the proposal variance is $1/25$ times the variance 
of the starting distribution of the bridge $N(0,\frac{1}{\phi_0})$; this is a choice that gave good acceptance 
probabilities over all $d$ bridging steps of the sampler.
Multinomial resampling was used when the effective sample size dropped below~$\frac{N}{2}$. 
We made 50 independent runs of the algorithm, and calculated the corresponding realisations of the log Ratio:    
\begin{equation}
\sum_{k=1}^{m^*(d)+1}\log\big(\tfrac{\gamma_{d,k}^N(1)}{\gamma_{d,k-1}^N(1)}\big)
\big/\log\big(\tfrac{\gamma_{d,k}(1)}{\gamma_{d,k-1}(1)}\big)
\label{eq:comp_criteria}
\end{equation}
(note that now the resampling times and their number are random)
and their sample variance.
This experiment was carried out for choices of dimension $d\in\{10,25,50\}$
and for both annealing schemes $\phi(s)$ and $\nu(s)$. 
The ratio of the obtained variances for the annealing sequence 
$\phi(s)$ over $\nu(s)$ are shown in Table~\ref{tab:rel_var}.
The results confirm
our theoretical findings above for the superiority of $\nu(s)$ over $\phi(s)$ 
based on the analytical expression for the asymptotic variance even for moderate $d$.

\begin{table}\centering
\begin{tabular}{|c|c|c|c|}
\hline
Dimension $d$ & 10  & 25 & 50 \\
\hline
Ratio of Variances (for $\phi(s)$ over $\nu(s)$) & 2.32  & 3.47 & 7.05 \\
\hline
\end{tabular}
\vspace{0.1cm}
\caption{The empirical variances (oven $50$ independent runs for the SMC sampler) of the log Ratio \eqref{eq:comp_criteria}  
using the annealing $\phi(s)$ over the corresponding variances for the annealing sequence $\nu(s)$.}
\label{tab:rel_var}
\end{table}

\begin{figure}
\centering
\subfigure[Annealing Schemes, $\phi_0=0.01$]{{\includegraphics[width=0.45\textwidth,height=6.5cm]{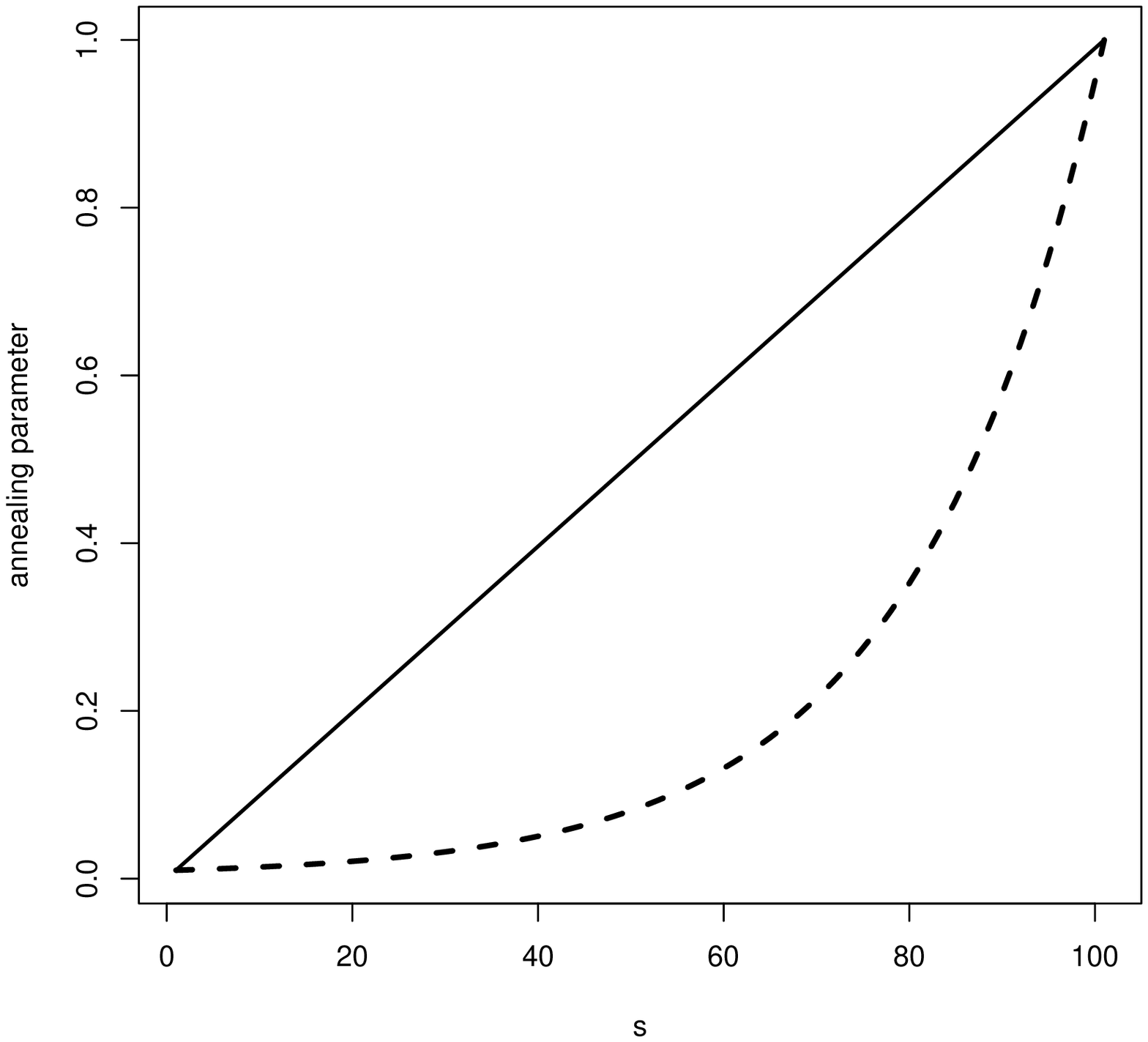}}}
\subfigure[Asymptotic Variances]{{\includegraphics[width=0.45\textwidth,height=6.5cm]{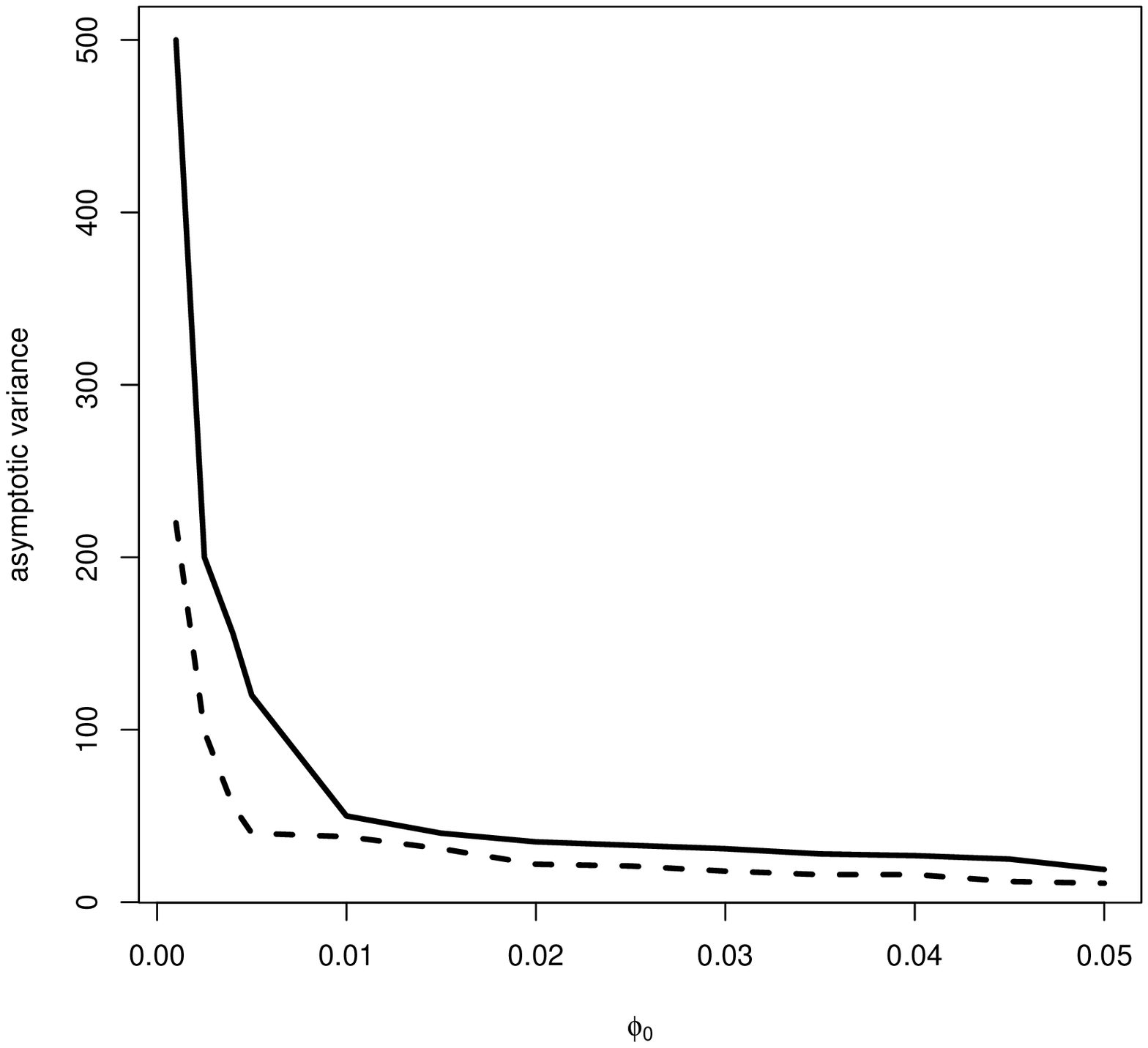}}}
\caption{Panel (a): Plot of the annealing parameters
$\phi(s)$ and $\nu(s)$ (broken line) defined in (\ref{eq:anna}) against time when $\phi_0=0.01$ and $\vartheta=5$.
Panel (b): A Plot of $\sigma^{2,\phi}_{0:1}$ and $\sigma^{2,\nu}_{0:1}$ (broken line) against $\phi_0$
for the case of exact sampling 
with $k_s(x,dx')\equiv \pi_s(dx')$ and the scenario (\ref{eq:avb}).}
\label{fig:asymp_prop}
\end{figure}

\subsection{Bayesian Linear Model}

We now consider the implications of main results
in the context of a Bayesian linear regression model (see \cite{holmes} for a book-length introduction as well as a wealth of practical applications).
This is a statistical model that associates a $p$-vector of responses, say $Y$, to 
a $p\times d$-matrix of explanatory variables, $X$, for some $p\geq 1$, $d\ge 1$.   
In particular:
$$
Y = X\beta + \epsilon
$$
where $\beta$ is a $d$-vector of unknown regression coefficients and 
$\epsilon\sim N_p(0,\mathbf{1}_p)$, with $\mathbf{1}_p$ the $p\times p$ identity matrix. A prior density on $\beta$ is taken as $N_d(0,\mathbf{1}_d)$
which yields a posterior density found to be the $d$-dimensional Gaussian $N_d((\mathbf{1}_d + X'X)^{-1}X'Y,
(\mathbf{1}_d + X'X)^{-1})$ where $X'$ denotes transpose. This is the target distribution for our SMC sampler.

The objective is to investigate the bound in Theorem \ref{theo:l2resampling} and the implications of 
Proposition~\ref{theo:prop_chaos}. Note that the target distribution is not of product structure here. The data-point tempering method (see Section \ref{sec:datapointtemp}) is also compared with annealing.
We consider the case $d=50$, $p=50$ with $N=10^3$; the data are all simulated. The annealing scheme $\nu$ in (\ref{eq:anna}) is adopted as well as the data-point tempering method with $\lfloor 10d/p\rfloor$ steps between the $p$ data point arrivals. 
Particles are propagated along the bridging densities via Markov kernels corresponding to Random-Walk Metropolis within Gibbs: 
the proposal for a univariate co-ordinate $x$ conditionally on the rest is $y = x + N(0,\frac{1}{16})$.
Dynamic resampling according to the ESS is employed (threshold 
$\frac{N}{2}$) as well as resampling at the last time step (see Theorem \ref{theo:l2resampling}).
For the annealing scheme the number of SMC steps is scaled as a multiple of $d$).
This increase at the number of time steps aims at illustrating the propagation of chaos (Proposition \ref{theo:prop_chaos}). 
 We fixed $d=50$
for computational cost considerations, but the SMC algorithms will easily stabilize for much larger $d$. 
\begin{table}\centering
\begin{tabular}{|c|c|c|c|}
\hline
Time-Steps & $d$  & $5d$ & $10d$ \\
\hline
Relative Error & 4.75  & 4.47 & 3.9 \\
\hline
\end{tabular}
\vspace{0.1cm}
\caption{The mean square error when estimating $\Exp[\,\beta_1\,|\,Y\,]$ for the (annealed) SMC sampler 
over 100 repeatitions relative to i.i.d.~sampling. 1000 particles are run and we resample at the final time step.
The error is calculated for different choices of number of time-steps for the SMC sampler.}
\label{tab:iid_var}
\end{table}

Each SMC method employed is repeatedly applied 100 times.
We calculate the mean square error for the estimation of $\Exp\,[\,\beta_1\,|\,Y\,]$ (analytically available here)
over the 100 replications and
we compare with the corresponding error under i.i.d.\@ sampling of the posterior of $\beta_1$; 
the results are reported in Table~\ref{tab:iid_var}.
In the table we can observe the increase in mean square error of the (annealed) SMC algorithm to i.i.d.~simulation. 
The increase here is not substantial as indicated by Theorem \ref{theo:l2resampling}, although one may need to take $d$ very large (and have an i.i.d.~target) before the bound in Theorem \ref{theo:l2resampling} is realized.
As the number of time-steps increases we can observe an improvement. This
is due to increase in diversity of the population, which improves the SMC estimate even when resampling at the end. For the data-point tempering method (the CPU time is roughly comparable with the case of $10d$ time-steps of the annealed SMC) the corresponding value of the relative mean square error is 7.5, which is slightly worse than the annealing scheme. In general, it is difficult to draw a definite conclusion on which scheme may be better.

\section{Filtering}\label{sec:filtering}

An important application of SMC methods is filtering. In the following 
we will look at the effect of dimension for several filtering algorithms.

\subsection{Set-Up}

Consider the discrete-time filtering problem; we have fixed observations $y_1,\dots,y_n$, with $y_k\in\mathbb{R}^{d_y}$
and a hidden Markov chain $X_{1},\dots,X_n$, with $X_k \in E^d$ such that the $Y_k$'s are conditionally independent of all other variables, given $X_k$.
We assume that the density w.r.t.~Lebesgue measure is
\begin{equation}
\label{eq:likelihood}
g(y_k|x_k) = \exp\bigg\{\sum_{j=1}^d h(y_k,x_{k,j})\bigg\}
\end{equation}
with $x_{k,j}\in E$ and $h:\mathbb{R}^{d_y}\times E \rightarrow\mathbb{R}$.
The hidden Markov chain is taken to be time-homogeneous with transition density w.r.t.~Lebesgue measure:
$$
F(x_k|x_{k-1}) = \prod_{j=1}^d f(x_{k,j}|x_{k-1,j})\ , \quad k\geq 1\ ,
$$
where $x_0=(x_{0,1},x_{0,2},\dots,x_{0,d})\in E^d$ is a given fixed point and $\int_{E}f(x'|x)dx'=1$. This is certainly a very specific model structure chosen, as in the previous part of the paper, for mathematical convenience. 
Clearly, our results depend on the given structure; it is certainly the case that for some other classes of state-space models standard SMC methods could stabilize with the dimension of the problem (as could be the case for instance 
when the lihelihood in (\ref{eq:likelihood}) involved only a few of the co-ordinates of $x$).

The objective in filtering is to compute for a $\pi$-integrable function $\varphi$:
\begin{equation}
\mathbb{E}\,[\,\varphi(X_n)\,|\,y_{1:n}\,] = \int_{E}\varphi(x) \pi_n(x_n|y_{1:n})dx_{n}
\label{eq:filter}
\end{equation}
where
$$
\pi_n(x_n|y_{1:n}) = \frac{\int_{E^{(n-1)d}} \prod_{k=1}^n g(y_k|x_k)F(x_k|x_{k-1}) dx_{1:n-1}}{\int_{E^{nd}} \prod_{k=1}^n g(y_k|x_k)F(x_k|x_{k-1}) dx_{1:n}}\ .
$$
It should be noted that one can re-write the filter via the standard prediction-updating formula:
\begin{align*}
\pi_n(x_n|y_{1:n}) & =  \frac{g(y_n|x_n) \pi_{n|n-1}(x_n|y_{1:n-1})}{p(y_n|y_{1:n-1})}\ ; \\
\pi_{n|n-1}(x_n|y_{1:n-1}) & =  \int_{E^d} F(x_n|x_{n-1})\pi_{n-1}(x_{n-1}|y_{1:n-1})dx_{n-1} \ ; \\
p(y_n|y_{1:n-1}) & =  \frac{\int_{E^{nd}} \prod_{k=1}^n g(y_k|x_k)F(x_k|x_{k-1}) dx_{1:n}}
{\int_{E^{(n-1)d}} \prod_{k=1}^{n-1} g(y_k|x_k)F(x_k|x_{k-1}) dx_{1:n-1}}\ .
\end{align*}

Typically, one cannot calculate \eqref{eq:filter}, so we resort to particle filtering. The most basic approach is to perform an SMC algorithm, which approximates the sequence of densities
$$
\pi(x_{1:n}|y_{1:n}) \propto \prod_{k=1}^n g(y_k|x_k)F(x_k|x_{k-1})
$$
by using the prior dynamics, characterised by $F$, as a proposal. This yields an un-normalized incremental weight
at time step $n$ of the algorithm which is $g(y_n|x_n)$. One can then resample or not.
We consider the scenario as the dimension of the state increases, when the data record is fixed (i.e.~we keep the time parameter $n$ and data fixed).

In reference to the works \cite{bengtsson,bickel,legland,snyder}, it is clear that the standard particle filter cannot be used in general to approximate filters with high-dimensional
states. 
An alternative, as mentioned for the data-point tempering scenario in Section \ref{sec:datapointtemp} (see also \cite{bickel,godsill}) is to insert an annealing SMC sampler between consecutive filtering steps, updating the entire trajectory $x_{1:n}\in E^{nd}$.

Assuming the i.i.d.~structure above, the model dynamics decompose over $d$ independent 
co-ordinates.
One could use MCMC kernels for the SMC samplers between arrivals of consecutive data-points, with each univariate kernel (the product of $d$ of them forming the complete kernel)
having invariant density:
$$
\pi_{p}^{(n)}(x_{1:n}) \propto \exp\Big\{\phi\big(\tfrac{p}{d}\big)\,h(y_n,x_n) + \log(f(x_n|x_{n-1}))+
\sum_{i=1}^{n-1} \{\,h(y_i,x_i) + \log(f(x_i|x_{i-1}))\,\}\,\Big\}\ ,
$$
for $x_{1:n}\in E^{n}$, $\phi(0)=0$, and $0\le p\le d$.
We write these marginal target densities at data-time $n$ on the continuum as $\pi^{(n)}_{s}$ with associated
 Markov kernels $k_{s}^{(n)}$ (which operate on spaces of increasing dimension). No resampling is added to the algorithm, but easily could be. The ESS (see \eqref{eq:ess_def} in Figure \ref{tab:SMC}) is denoted $\textrm{ESS}_{(0,d)}(n,N)$.
It is a straight-forward application of Theorem~3.1 of \cite{beskos} 
to get that, under assumptions (A\ref{hyp:A}-\ref{hyp:B}) for the kernels $k_{s}^{(n)}$, $n\geq 1$, and  
the condition that  for each $n\geq 1$,  $h(y_n;x_n)\in\mathcal{B}_b(E)$
(with the associated solution to Poisson's equation written $\hat{g}_{s}^{(n)}$),
for any fixed $N>1$, $n\geq 1$, $\textrm{ESS}_{(0,d)}(n,N)$ converges in distribution to 
$$
\frac{[\,\sum_{i=1}^N e^{Z^{i}}\,]^2}{\sum_{i=1}^N e^{2 Z^{i}}}
$$ 
where $Z^{i} \stackrel{i.i.d.}{\sim}N(0,\sigma^2)$ with
$$
\sigma^2 = \sum_{k=1}^n \int_{0}^1 \pi_{\phi(u)}^{(k)}((\hat{g}_{\phi(u)}^{(k)})^2-
k_{\phi(u)}^{(k)}(\hat{g}_{\phi(u)}^{(k)})^2)\bigg[\frac{d\phi(u)}{du}\bigg]
d\phi(u)\ .
$$
In particular,
\begin{equation}
\lim_{d\rightarrow\infty}
\mathbb{E}\,\big[\,\textrm{ESS}_{(0,d)}(n,N)\,\big]
=\mathbb{E}\bigg[\frac{[\,\sum_{i=1}^N e^{Z^{i}}\,]^2}{\sum_{i=1}^N e^{2 Z^{i}}}\bigg]\ .
\label{eq:limitingess}
\end{equation}
\noindent Thus the cost for this algorithm to be stable as $d\rightarrow \infty$, is $\mathcal{O}(n^2 d^2 N)$. 
This result is not surprising; one updates the whole state-trajectory and the stability proved in \cite{beskos} is easily imported into the algorithm. However, this algorithm is not online and will be of limited practical significance unless one has access to substantial computational power. 
The result is also slightly misleading:
it assumes that the MCMC kernels have a uniform mixing with respect to the time parameter of the HMM (see condition 
A\ref{hyp:A}). This is unlikely to hold unless one increases the computational effort, associated to the MCMC kernel, with $n$.

One apparent generalization would be to use SMC samplers at each data-point time to sample from the annealed smoothing densities, except freezing the first $n-1$ co-ordinates; it is then easily seen that one does not have to
store the trajectory. However, one can use the following intuition as to why this procedure will not work well (the following is based upon personal communication with Prof.~A.~Doucet). In the idealized scenario, one samples exactly from the final target density of the SMC sampler. In this case, the final target density is exactly the conditionally optimal proposal (see \cite{doucet}) and the incremental weight is: 
$$
\int_{E^d}g(y_n|x_n)F(x_n|x_{n-1})dx_n = 
\prod_{j=1}^d \int_{E}e^{h(y_n,x_{n,j})}f(x_{n,j}|x_{n-1,j})dx_{n,j}\ ,
$$
which will typically have exponentially increasing variance in $d$. We conjecture that similar issues arise for advanced SMC approaches such as \cite{chorin}. 

\subsection{Marginal Algorithm}

Due to the obvious instability of the above procedure, we consider another alternative, presented e.g.\@ in \cite{poyiadjis}.
This algorithm would proceed as follows. When targeting the filter at time 1, one adopts an SMC sampler, which as discussed before, will stabilize with the dimension. Then at subsequent time-steps, to consider the initial density of the SMC sampler:
\begin{equation}
\sum_{l=1}^N \overline{w}_{d}^{l,(n-1)}\prod_{j=1}^d f(x_{0,j}^{l',(n)}|x_{d,j}^{l,(n-1)})
\label{eq:first_density}
\end{equation}
where we have defined: 
$$
\overline{w}_{0:d-1}^{l,(n-1)} \propto \exp\Big\{\,\tfrac{1}{d}\sum_{j=1}^d\sum_{i=0}^{d-1} h(y_{n-1},x_{i,j}^{l,(n-1)})
\,\Big\}\ ;\quad \sum_{l=1}^N \overline{w}_{d}^{l,(n-1)} = 1 \ ,
$$
as the normalized weight.
One could then resample and apply SMC samplers on the sequence of target distributions (e.g.~with $n=2$)
\begin{equation}
\pi_k^{(n)}(x_{1:d}) \propto \exp\Big\{\,\phi\big(\tfrac{k}{d}\big)\sum_{j=1}^d h(y_n,x_{j})\,\Big\}
\Big[\,\sum_{l=1}^N \prod_{j=1}^d f(x_{j}|\check{x}_{d,j}^{l,(n-1)})\,\Big]\ , \quad k\in\{0,\dots,d\}\ ,\label{eq:pi_approx}
\end{equation}
where $\check{x}_{d,j}^{l,(n-1)}$ is the resampled particle when sampling from \eqref{eq:first_density}.

Using simple intuition, one might expect that this SMC algorithm may stabilize as the dimension grows. For example, if one could sample exactly from  $\pi_p^{(d)}(x_{1:d})$,
then the importance weight is exactly 1; there is no weight degeneracy problem: As proved in \cite{beskos}, under ergodicity assumptions, the SMC sampler will asymptotically produce a sample from the final target density. 

However,
the following result suggests that the algorithm will collapse unless
the number of particles grows exponentially fast with the dimension. Consider 
the case with $N_d$ particles, where $N$ depends on $d$ here,
and these are samples exactly from the previous filter; denote the samples $\check{X}_{1:d}^{1:N_d}$. Conditionally upon $\check{X}_{1:d}^{1:N_d}$,
sample $X_{1:d}^{1:N_d}$ exactly from the approximation $\pi_d^{(n)}$ \eqref{eq:pi_approx}. This presents the most optimistic scenario one could hope for. Denote below, $\overline{\varphi}_n(x)
=\varphi(x)-\pi_n(\varphi)$. We have the following result.

\begin{prop}\label{prop:marginal_doesnt_seem_to_work}
Consider the algorithm above so that for any $(x,x')\in E^2$ we have  $f(x|x')\in(\underline{f},\overline{f})$
for constants
$0<\underline{f}<\overline{f}<\infty$
and for any $(y,x)\in\mathbb{R}^{d_y}\times E$, $h(y|x)\in (\underline{h},\overline{h})$ for
$0<\underline{h}<\overline{h}<\infty$. Suppose $\varphi\in\mathcal{B}_b(E)$. Then, there exist an $M<\infty$ and $\kappa>1$ such that for any $d\geq 2$,
$j\in\{1,\dots,d\}$ and $N_d\geq 1$ we have
$$
\mathbb{E}\,\Big[\,\big(\tfrac{1}{N_d}\sum_{i=1}^{N_d}\overline{\varphi}_n(X_{j}^i)\,\big)^2\,\Big] \leq \frac{M\kappa^d}{N_d}\ .
$$
\end{prop}

\begin{proof}
Throughout the proof, write $\pi^{(n)}_d$ as the approximated filter at time $n$. 
Then we have the simple decomposition:
$$
\Exp\,\Big[\,\big(\,\tfrac{1}{N_d}\sum_{i=1}^{N_d}\varphi(X_{j}^i)-\pi_n(\varphi)\,\big)^2\,\Big]
= 
\Exp\,\Big[\,\big(\,\frac{1}{N_d}\sum_{i=1}^{N_d}\varphi(X_{j}^i)- \pi^{(n)}_d(\varphi)
+ \pi^{(n)}_d(\varphi)-\pi_n(\varphi)\,\big)^2\,\Big]\ .
$$
On applying the $C_2-$inequality, we can decompose the error into the one of the Monte Carlo error of approximating expectations w.r.t.~$\pi^{(n)}_d$
and that of approximating the filter. 

Consider the first error. Conditioning on the i.i.d.~samples drawn from the filter at time $n-1$, one may apply the Marcinkiewicz-Zygmund inequality and use the i.i.d.~property of the algorithm to obtain  the upper bound:
$$
\big(\tfrac{M}{N_d}\big)\,\Exp\bigg[\frac{\frac{1}{N_d}\sum_{i=1}^{N_d}f(\varphi^2e^{h})(\check{X}_{j}^{i})\prod_{l\neq j}f(e^h)(\check{X}_{l}^{i})}{\frac{1}{N_d}\sum_{i=1}^{N_d}\prod_{l=1}^df(e^h)(\check{X}_{l}^{i})} - 
\bigg\{\frac{\frac{1}{N_d}\sum_{i=1}^{N_d}f(\varphi e^{h})(\check{X}_{j}^{i})\prod_{l\neq j}f(e^h)(\check{X}_{l}^{i})}{\frac{1}{N_d}\sum_{i=1}^{N_d}\prod_{l=1}^df(e^h)(\check{X}_{l}^{i})}\bigg\}^2
\bigg]\ .
$$
As $\varphi$ and $h$ are bounded, it is easily seen that the function in the expectation is uniformly bounded in $d$ and hence that one has an upper-bound of the form $M/N_d$.

Now to deal with the second error, this can be written:
$$
\Exp\,\bigg[\,\bigg(\,\frac{\frac{1}{N_d}\sum_{i=1}^{N_d}f(\varphi e^{h})(\check{X}_{j}^{i})\prod_{l\neq j}f(e^h)(\check{X}_{l}^{i})}{\frac{1}{N_d}\sum_{i=1}^{N_d}\prod_{l=1}^df(e^h)(\check{X}_{l}^{i})}
-
\frac{\pi_{n-1}f(\varphi e^h)\pi_{n-1}f(e^h)^{d-1}}{\pi_{n-1}f(e^h)^{d}}
 \,\bigg)^2\,\bigg]\ .
$$
The bracket can be decomposed into the form:
\begin{align*}
\bigg(\,\frac{\frac{1}{N_d}\sum_{i=1}^{N_d}f(\varphi e^{h})(\check{X}_{j}^{i})\prod_{l\neq j}f(e^h)(\check{X}_{l}^{i})}{\{\pi_{n-1}f(e^h)^{d}\}\frac{1}{N_d}\sum_{i=1}^{N_d}\prod_{l=1}^df(e^h)(\check{X}_{l}^{i})}
\bigg)\bigg(\pi_{n-1}f(e^h)^{d}-
\tfrac{1}{N_d}\,\sum_{i=1}^{N_d}\prod_{l=1}^df(e^h)(\check{X}_{l}^{i})\bigg) +\\ 
\frac{1}{\pi_{n-1}f(e^h)^{d}}
\bigg(\frac{1}{N_d}\sum_{i=1}^{N_d}f(\varphi e^{h})(\check{X}_{j}^{i})\prod_{l\neq j}f(e^h)(\check{X}_{l}^{i})-\pi_{n-1}f(\varphi e^h)\pi_{n-1}f(e^h)^{d-1}\,\bigg)\ .
\end{align*}
Applying the $C_2-$inequality again, we can break up the two terms. Using the lower bound on $h$ and upper-bounds on $\varphi$ and $h$ the $\mathbb{L}_2$-error
of the first term is upper-bounded by:
$$
\tfrac{\|\varphi\|_{\infty}^2e^{2\overline{h}}}{\pi_{n-1}f(e^h)^{2d}e^{2\underline{h}}}\,
 \Exp\,\bigg[\,\bigg(\,\pi_{n-1}f(e^h)^{d}-\tfrac{1}{N_d}\sum_{i=1}^{N_d}\prod_{l=1}^df(e^h)(\check{X}_{l}^{i})\,\bigg)^2\,\bigg]
= \tfrac{\|\varphi\|_{\infty}^2e^{2\overline{h}}}{N_d \pi_{n-1}f(e^h)^{2d}e^{2\underline{h}}}
\mathbb{V}\textrm{ar}\,\Big[\,\prod_{l=1}^df(e^h)(\check{X}_{l}^{1})\,\Big]\ .
$$
As the variance on the R.H.S.~is easily seen to be equal to $\pi_{n-1}(f(e^h)^{2d})-\pi_{n-1}(f(e^h))^{2d}$
one yields the upper-bound:
$$
\frac{M}{N_d}\bigg[\frac{\pi_{n-1}(f(e^h)^{2d})}{\pi_{n-1}(f(e^h))^{2d}}-1\bigg]\ .
$$
For the second term one can follow similar arguments to yield the upper-bound:
$$
\frac{1}{N_d\pi_{n-1}[f(e^h)]^{2d}}\,\Big[\,\pi_{n-1}[f(\varphi e^h)|^2]\pi_{n-1}[f(e^h)^2]^{(d-1)}
- 
\pi_{n-1}[f(\varphi e^h)]^2\pi_{n-1}[f(e^h)]^{2(d-1)}\,
\Big]
$$
from which one can easily conclude.
\end{proof}

\begin{rem}
On inspection of the proof, it is easily seen that one can 
write the error as
$
\tfrac{M}{N_d}(1+\kappa^d)
$
which represents two sources of error. The first is the Monte Carlo error due to estimating the marginal expectation w.r.t.~the approximation. This appears to be controllable for any $N_d$ converging to infinity. The second source of error is in approximating the filter, which seems to require a number of particles which will increase exponentially in the dimension; this is the drawback of this algorithm. 
\end{rem}

\begin{rem}
We remark that this is only an upper-bound, but we can be even more precise; if one considers the relative $\mathbb{L}_2$-error of the estimate of $p(y_n|y_{1:n-1})$ then this is equal to
$$
\frac{1}{N_d}\,\bigg[\,\frac{\pi_{n-1}(f(e^h)^2)^d}{\pi_{n-1}(f(e^h))^{2d}}-1\,\bigg]
$$
which will explode in the dimension, unless $N_d$ grows at an exponential rate. This is in contrast to the SMC sampler case in Section \ref{sec:main_res}, where one
can obtain an estimate of the normalizing constant, whose relative $\mathbb{L}_2$-error
stabilizes for any $N$.
\end{rem}

\subsection{Approximate Bayesian Computation (ABC)}\label{sec:abc_analysis}

In this section we consider SMC methods in the context of an ABC filter - an approximate filtering scheme which is of practical interest when evaluation of the likelihood function in the state-space model is intractable. We note in passing the connection of the ABC methods to the ensemble Kalman filter \cite{nott}, a full treatment of the latter is well beyond the scope of the present work. 

The idea of this approach, which is primarily adopted when:
\begin{itemize}
\item{The function $g(y|x_{1:d})$ is intractable, that is, one cannot evaluate it point-wise.}
\item{It is possible to simulate from $g(\cdot|x_{1:d})$ for any $x_{1:d}\in E^d$.}
\end{itemize}
In this scenario, standard SMC methods can be used to sample from an approximation of the smoothing density, of the form (for some $\epsilon>0$):
\begin{equation}
\pi_{\epsilon}(x_{1:n},u_{1:n}|y_{1:n}) \propto \prod_{k=1}^n \mathbb{I}_{\{u_k:|y_k-u_k|<\epsilon\}}(u_k)\, g(u_k|x_k)\,\overline{f}(x_k|x_{k-1})\label{eq:abc_approx}\ .
\end{equation}
Here, the idea is to sample, at each time-point, pseudo-data $u_k\in\mathbb{R}^{d_y}$; the density is non-zero when all of the simulated pseudo data lie within $\epsilon$ of the observed data (in $\mathbb{L}_1$-distance). Adopting an SMC algorithm with proposals $g\otimes \overline{f}$ yields an un-normalized incremental weight of the form $\mathbb{I}_{\{u_k:|y_k-u_k|<\epsilon\}}(u_k)$, which circumvents the evaluation of $g$. In the context of high-dimensional models, as discussed here, the SMC algorithm will collapse using the approaches in the previous sections. However, when $d_y$ is not large, one would expect that indeed, the SMC approximation of the ABC filter should be reasonably stable (in some sense). We quantify this with the following result.

We will assume here conditions (A1-A3) of \cite{jasra}. In particular, that:
\begin{equation*} 
\sup_{x_{1:d}}\|g(\cdot,x_{1:d})\|_{\infty}=e^{d\sup_{x_{1:d}}\|h(\cdot,x_{1:d})\|_{\infty}}\ ; 
\quad \sup_{x_{1:d}}|g(y,x_{1:d})-g(u,x_{1:d})| \leq M\,|y-u|\ .
\end{equation*}
The latter assumption will typically only hold when $E\subset \mathbb{R}$ with $E$ being compact. We consider only the scenario where no resampling is performed.
The expectation with respect to the SMC algorithm conditioned on the fixed data (which is suppressed from the notation) is written as $\mathbb{E}$.
Also, we write simply $\mathbb{I}_{A_{y_{1:n},\epsilon}}$ in the place of  $\mathbb{I}_{\{u_{1:n}^{1:N}:\sum_{j=1}^N \prod_{k=1}^n \mathbb{I}_{\{u_k\,:\,|y_k-u_k|<\epsilon\}}(u_k^j)>0\}}$.

\begin{prop}\label{prop:abc}
Given the set-up above, one has that for any $n\geq 1$, $p\geq 1$, $\varphi\in\mathcal{B}_b(E)$, $\epsilon>0$ there exists an $M(n,p,\varphi,\epsilon)>0$, and for $d\geq 1$ there exists $\kappa_n(d,\varphi)>0$ which does not depend upon $p,\epsilon$ such that for any $N\geq 1$:
\begin{equation}
\mathbb{E}\bigg[\,\bigg|\,\sum_{i=1}^N \frac{\prod_{k=1}^n \mathbb{I}_{\{u_k:|y_k-u_k|<\epsilon\}}(u_k^i)}
{\sum_{j=1}^N \prod_{k=1}^n \mathbb{I}_{\{u_k:|y_k-u_k|<\epsilon\}}(u_k^j)} \overline{\varphi}_n(X_{n,1}^i)\,\bigg|^p\,\mathbb{I}_{A_{y_{1:n},\epsilon}}\,\bigg]^{1/p} \leq \tfrac{M(n,p,\varphi,\epsilon)}{\sqrt{N}} + \kappa_n(d,\varphi)\epsilon
\label{prop:abc_eq}
\end{equation}
where $\kappa_n(d,\varphi)$, as $d\rightarrow\infty$,  converges to zero or diverges to infinity.
\end{prop}

\begin{proof}
Write
$$
\pi_{n,\epsilon}(\varphi):= \int \varphi(x_{n,1})\pi_{\epsilon}(x_{1:n},u_{1:n}|y_{1:n})dx_{1:n}du_{1:n}\ .
$$
Then, one can add and subtract this term in the $|\cdot|^p$ and apply Minkowski, leading to:
$$
\Exp\,\Big[\,\Big|\sum_{i=1}^N \frac{\prod_{k=1}^n \mathbb{I}_{\{u_k:|y_k-u_k|<\epsilon\}}(u_k^i)}
{\sum_{j=1}^N \prod_{k=1}^n \mathbb{I}_{\{u_k:|y_k-u_k|<\epsilon\}}(u_k^j)} \varphi(X_{n,1}^i)-\pi_{n,\epsilon}(\varphi)\Big|^p\mathbb{I}_{A_{y_{1:n},\epsilon}}\,\Big]^{1/p}
+ |\pi_{n,\epsilon}(\varphi)-\pi_{n}(\varphi)|\ .
$$
The first term is easily dealt with using standard proof techniques in Monte Carlo computation; see for example part of the proof of Theorem 3.3.~of \cite{beskos}.
Hence we need only treat the bias term. Following the arguments of Theorem 1 of \cite{jasra}, one can obtain that an upper bound on the bias is:
$$
\kappa_n(d,\varphi) = \tfrac{\epsilon\|\varphi\|_{\infty}}{\big(\int_{\mathbb{R}^2} e^{h(y_n,x)}f(x|x')\pi_{n-1,1}(x'|y_{1:n-1})
dxdx'\big)^d}\,\,\big(\,2L + \kappa_{n-1}(d,\varphi)e^{d\sup_{x}\|h(y_n,x)\|_{\infty}}\,\big)
$$
where $\pi_{n-1,1}$ is the filter at time $n-1$ marginalized to its first component and $\kappa_0(d,\varphi)=0$. 
\end{proof}

\begin{rem}
The above result is of interest for high-dimensional filtering. Essentially it establishes that the SMC approximation of the ABC approximation is stable for any $d\geq 1$, with computational cost of $\mathcal{O}(Nd)$; this is the first term on the R.H.S.~of \eqref{prop:abc_eq}. However, the deterministic component of the ABC approximation of the filter is likely to deteriorate as $d\rightarrow\infty$. For any $n\geq 1$ the sequence $(\kappa_n(d,\varphi))_{d\geq 1}$ is likely to diverge, as for example when $n=1$ it is proportional to
$$
\bigg(\int_{E} e^{h(y_1,x_1)} f(x_1|x_{0,1})dx_1\bigg)^{-d}.
$$
Whilst this is only  an upper-bound, we will see in Section \ref{sec:abc_simos} that the error seems to increase with $d$ in simulations.
\end{rem}

\begin{rem}
Due to the link between ABC and EnKF \cite{nott} and the bias of the EnKF \cite{legland1}, we conjecture that the EnKF will be subject to a similar behavior as for ABC (for non-linear models). That is, one can numerically approximate the approximation of the filter in high dimensions, but that the approximation collapses as the dimension of the state grows.
\end{rem}

\subsection{Numerical Example: Linear Gaussian State-Space Model}\label{sec:abc_simos}

In the following example we consider the ABC approximation error of linear Gaussian state-space model, $k\geq 1$:
\begin{align*}
Y_k & =  H X_k + V_k\ ; \\
X_k & =  X_{k-1} +  W_k\ ,
\end{align*}
where $d_y=1$, $H=(1,\dots,1)$ is a $1\times d$ vector, $V_k\stackrel{\textrm{i.i.d.}}{\sim}N_1(0,1)$, $X_0=(0,\dots,0)'$ and $W_k\stackrel{\textrm{i.i.d.}}{\sim}N_d(0,\mathbf{1}_d)$.
200 data points are generated from the model.

We run the SMC-based ABC algorithm in \cite{jasra} with the parameter $\epsilon$
(see (\ref{eq:abc_approx}) for the approximation of the smoothing density) fixed at 5, with $d\in\{10,40,200\}$. The first moment in the first dimension is estimated and the quantity on the L.H.S.~of \eqref{prop:abc_eq} (associated to this) is estimated with $N=1000$ with 50 repeats. The results can be observed in Figure \ref{fig:abc_err}.
where  the estimate of the L.H.S.~of \eqref{prop:abc_eq} over times 1 to 200 for $d=200$ against $d=10$ (black) and $d=40$ (broken blue) is plotted. It appears that the error grows with $d$. It is remarked that in general as the model is not i.i.d.~(in dimension) one cannot guarantee a uniform per-time step increase in the error. However, the relative increase in the error, w.r.t.~dimension, is consistent with our empirical experience with applying the algorithm.

\begin{figure}[h]\label{fig:abc_err}
\centering
{{\includegraphics[width=0.8\textwidth,height=5.5cm]{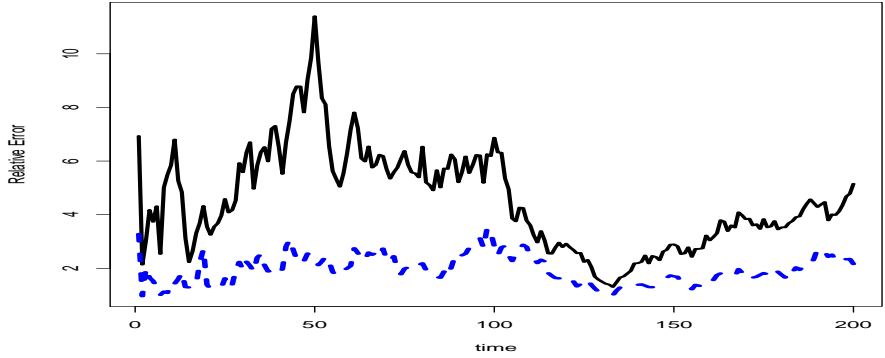}}}
\caption{A Plot of the Relative $\mathbb{L}_2$-Error of the ABC Filter Algorithm. The black line is the 
$\mathbb{L}_2$-error for estimating the first moment in the first dimension for 200 dimensions against 10 dimensions over 200 time steps. The broken blue line is similar except for 200 dimensions versus 40 dimensions.}
\end{figure}

\subsection{Discussion on High-Dimensional Filtering}

One of the motivations of our work was to investigate the issue of stability in high dimensions of SMC algorithms for filtering. As can be seen, there is still much scope for future work and this issue is far from resolved. In particular and in relation to this issue (as established in Section \ref{sec:abc_analysis}) what is currently missing in the literature is a concerted effort from probabilists, statisticians and applied mathematicians on the analysis of algorithms used in data assimilation (at least one exception is \cite{legland1}).

In relation to the above discussion,
one potentially very fruitful starting point, is the filtering of the Navier stokes equation; see \cite{brett}. In this context, one is given the 2-dimensional Navier-Stokes equation on a torus and the objective is to infer the initial condition of the PDE, given access to noisy data. In mathematical terms,
given a Gaussian prior on $x$ the initial condition, on Hilbert space $\mathcal{H}$
one seeks to deal with the density of the filter w.r.t.~the prior (see \cite[Theorem 3.2]{brett})
$$
\pi_k(x|y_{1:k}) \propto \exp\{-\Phi_k(y_{1:k},x)\}
$$
where $x\in\mathcal{H}$ and $\Phi_k:\mathbb{R}^k\times\mathcal{H}\rightarrow\mathbb{R}$ a potential function
associated to observed data $y_{1:k}\in\mathbb{R}$. This is an example of trying to filter a state with deterministic dynamics, albeit in infinite dimensions. Given the analysis in Section \ref{sec:datapointtemp}, it is likely that, by defining a data-point tempering SMC algorithm on the infinite dimensional space, one can consider the associated finite-dimensional algorithm (with state-vector in $\mathbb{R}^d$) and establish the stability as $d$ grows. This is assuming one has sufficiently mixing MCMC kernels; see \cite{pillai} for some ideas. This issue is a subject of current research jointly with Prof.~A.~Stuart.

More generally, one is interested in the issue of when the dynamics of the hidden state are Markovian. As noted here and in more details in \cite{bickel},
the underlying structure of the state-space model is important for establishing some sort of stability in high dimensions of a numerical filtering algorithm, SMC or otherwise. Akin to the standard problem of filter-stability in time (e.g.~\cite{vanhandel}), perhaps considerable research is needed with regards to dimension of the deterministic (true) filter, before a full analysis of numerical algorithms can be undertaken. However, it is not obvious how that analysis can be undertaken.

\section{Summary}\label{sec:disc}
 
In this paper we have considered the stability of SMC methods in high-dimensions.
In particular: the $\mathbb{L}_2$-error of marginal estimates, the $\mathbb{L}_2$-relative error of the normalizing constants and propagation of chaos properties.
The stability of some SMC-based filtering algorithms have also been investigated.
Some directions for future work are as follows. 

Firstly, in the context of normalizing constants, one direction is the consideration of rare events problems. Following  \cite{cerou1}, it is possible to obtain computational complexity results for some rare events problems. However, one can pose some rare-events problems in terms of the dimensionality.
Our results would, in many cases, not apply to this scenario and an extension to this case is important. In the case where one uses SMC samplers to sample from `twisted' target densities (see \cite{johansen}) the analysis adopted here can be applied. However, one would still need to verify that the path-sampling-based estimate will stabilize as $d$ grows.

Secondly, for normalizing constants, we have only considered the relative $\mathbb{L}_2$-error. It would be of interest to consider e.g.~logarithmic efficiency or higher-order errors. In addition, we have only considered one particular important functional that grows with $d$. More generally, when one can perform estimation with direct Monte Carlo, with a cost which is less than exponential in $d$, is it possible to do this also with SMC methods?
 
Thirdly, and rather importantly, is it possible to find any online SMC algorithm to solve the filtering problem in general, whose cost does not increase exponentially in the dimension? At present, our only suggestion is the accept/reject scheme in \cite{delmoral}. We are currently investigating the stability properties of this algorithm. It could be that in general, as noted above, one cannot obtain a stability result as in \cite{beskos}.

Finally, one could considerably weaken the the hypotheses made in this article. Given the number of exponential moments that we need to treat, it seems that multiplicative drift conditions \cite{kont} could be adopted; see \cite{whiteley1}.

\subsubsection*{Acknowledgements}

We would like to thank Arnaud Doucet and Anthony Lee for some useful discussions on this work. The work of Dan Crisan was partially supported by the EPSRC Grant No: EP/H0005500/1.
The third and fourth authors acknowledge assistance from a LMS research in pairs grant. The third author is supported by an ministry of education grant.

\appendix

\section{Proofs}


\subsection{Preliminary Results}
\label{app:0}

We summarize in  Lemmas \ref{lem:CLT} and \ref{lem:growth} below some results required in the proofs
obtained in \cite{beskos} or implied directy from results in that paper. Recall the definition 
of $G_{k,j}^{i}$ from (\ref{eq:defineGG}).
\begin{lem}[$G$-Asymptotics]
\label{lem:CLT}
Assume (A\ref{hyp:A}-\ref{hyp:B}) and $g\in\mathcal{B}_b(E)$.
\begin{itemize}
\item[i)] Under the starting distribution $X^{1:N}_{l_d(t_{k-1}(d)),1:d}\sim
\pi_{t_{k-1}(d)}^{\otimes Nd}$ we have that:
\begin{equation*}
\tfrac{G_{k,j}^{i}}{\sqrt{d}} \Rightarrow N(0,\sigma^{2}_{t_{k-1}:t_{k}})\ ;\quad  
\tfrac{1}{d}\,\sum_{j=1}^{d} G_{k,j}^{i}\Rightarrow N(0,\sigma^{2}_{t_{k-1}:t_{k}})\ .
\end{equation*}
%
%
\item[ii)] We have that $|\,\Exp_{\check{X}_{l_d(t_{k-1}(d)),j}^{i}}\,[\,G_{k,j}^{i}\,]\,|\le M$,
and for any $p\ge 2$:
\begin{equation*}
 \Exp_{\check{X}_{l_d(t_{k-1}(d)),j}^{i}}\,|G_{k,j}^{i}|^{p} \le M\, d^{\frac{p}{2}\vee 1}\  .
\end{equation*}

\item[iii)] Under either 
$\pi_{t_{k-1}(d)}^{\otimes Nd}$ as in i) or the actual particle distribution we have that:
\begin{equation*}
\Exp\,[\,e^{\frac{c}{d}\,\sum_{j=1}^{d} G_{k,j}^{i}}\,]\rightarrow 
\Exp\,[\,e^{c\,N(0,\sigma^{2}_{t_{k-1}:t_k})}\,]\equiv e^{\frac{1}{2}\,c^2\,\sigma^{2}_{t_{k-1}:t_k}}\ .
\end{equation*}
\item[iv)] We have that:
\begin{gather*}
\tfrac{1}{d}\sum_{j=1}^d\Exp_{\check{X}_{l_d(t_{k-1}(d)),j}^{i}}
[\,G^i_{k,j}\,]\rightarrow 0 \ ,\quad \textrm{in }\,\,\mathbb{L}_1 \ ;
 \\
\tfrac{1}{d^2}\,\sum_{j=1}^{d}\Exp_{\check{X}_{l_d(t_{k-1}(d)),j}^{i}}\big[\,\big(\,G_{k,j}^{i}-
\Exp_{\check{X}_{l_d(t_{k-1}(d)),j}^{i}}\,[\,G_{k,j}^{i}\,]\,\big)^2\,\big] \rightarrow \sigma^{2}_{t_{k-1}:t_k}
\ ,\quad \textrm{in }\,\,\mathbb{L}_1 \ .
\end{gather*}
\end{itemize}
\end{lem}
\begin{proof}\hspace{0.2cm} \vspace{-0.8cm}\\ 
\begin{itemize}
 \item[i)] Both weak limit follows from the proof of 
Theorem 3.2 of~\cite{beskos}. Notice, that a minor difference is that instead of the \emph{fixed} times 
$\phi_0$ and $1$ considered in Theorem~3.2 of \cite{beskos} we now sum terms 
between the \emph{varying} time instances $t_{k-1}(d)$ and  $t_{k}(d)$. However, the proof for this case follows 
trivially from the proof for the fixed times due to the limits $t_{k-1}(d)\rightarrow t_{k-1}$ 
and  $t_{k}(d)\rightarrow t_k$. 
\item[ii)] All these results follow directly from  Theorem A.1 of \cite{beskos}.
\item[iii)]
This  follows from the CLT's in parts i) and ii) and the uniform integrability result
obtained in Lemma~\ref{lem:ui_exp}. 
\item[iv)] The first result corresponds to Proposition C.4 of \cite{beskos}. The second result is shown in the 
proof of Theorem 4.1 of \cite{beskos}.
\vspace{-0.5cm}
\end{itemize}
\end{proof}
\begin{lem}
\label{lem:growth}(Convergence of Marginal Laws)
Assume (A\ref{hyp:A}-\ref{hyp:B}) and $g\in\mathcal{B}_b(E)$. Then we have:
\begin{itemize}
 \item[i)] For a sequence of times $s(d)\in(\phi_0,1)$ with $t_{k-1}(d)<s(d)$ and $s(d)\rightarrow s\in(t_{k-1},1)$ 
and the collection of time steps $u(d)=(l_d(t_{k-1}(d))+1):l_d(s(d))$ we have that as $d\rightarrow\infty$:
\begin{gather*}
\|k_{u(d)}(\check{X}_{l_d(t_{k-1}(d)),1}^{i})-\pi_{t_{k-1}(d)}k_{u(d)}\|_{tv} \rightarrow 0\ ,\quad
 \textrm{in }\,\,\, \mathbb{L}_{1}\ ; \\
\|\pi_{t_{k-1}(d)}k_{u(d)}-\pi_{s(d)}\|_{tv} \rightarrow 0\ .
\end{gather*}
\item[ii)] For a sequence of times $s(d)\in(\phi_0,1)$ with $t_{k-1}(d)<s(d)$ and $s(d)\rightarrow s\in(t_{k-1},1)$ 
and the collection of time steps $u(d)=l_d(t_{k-1}(d)):l_d(s(d))$ we have that:
\begin{equation*}
\big(\,\overline{w}^{1:N}_{u(d)},\,X^{1:N}_{l_d(s(d)),1}\,\big)  \Rightarrow 
\big(\,\tfrac{e^{Z^{1:N}}}{\sum_{l=1}^{N}e^{Z^l}},\, Y^{1:N} \,\big)
\end{equation*}
where $\{Z^i\}_{i=1}^{N}$ are i.i.d.\@ copies from $N(0,\sigma^2_{t_{k-1}:s})$ and, independently,  
$\{Y^i\}_{i=1}^{N}$ are  i.i.d.\@ copies from $\pi_s$.
%
%
\end{itemize}
\end{lem}
\begin{proof}
\hspace{0.2cm} \vspace{-0.8cm}\\ 
\begin{itemize}
\item[i)] The first result follows by the proof of Proposition C.4 of \cite{beskos}; the second result from Proposition A.1 
of \cite{beskos}.
\item[ii)] The weak convergence of the weights is analytically illustrated in the proof of 
Theorem 4.1 of~\cite{beskos}. The weak convergence of the positions of the Markov chain is proven 
in Proposition A.1 of~\cite{beskos}. The independence between the $Z^{1:N}$ and $Y^{1:N}$ 
limiting variables follows trivially 
from the fact that any single co-ordinate has a vanishing effect on the weights as $d\rightarrow \infty$.
\vspace{-0.5cm}
\end{itemize}
\end{proof}
%
%
%

%
%
%

\subsection{$\mathbb{L}_2$-Error}\label{app:l2proof}

\begin{proof}[Proof of Theorem \ref{theo:l2resampling}]
We begin by noting that, due to exhangeability of the particles:
\begin{equation}
\Exp\Big[\,\Big(\tfrac{1}{N}\sum_{i=1}^N[\,\varphi(\check{X}_{d,1}^i)-\pi(\varphi)\,]\Big)^2\,\Big]
= 
\tfrac{1}{N} \,\Exp\,[\,\{\,\overline{\varphi}(\check{X}_{d,1}^1)\,\}^2\,] + \big(\tfrac{N-1}{N}\big)\,
\Exp\,[\,\overline{\varphi}(\check{X}_{d,1}^1)\,\overline{\varphi}(\check{X}_{d,1}^2)\,]
\label{eq:first_decomp_l2}
\end{equation}
where we have set $\overline{\varphi}(x)=\varphi(x)-\pi(\varphi)$.
Starting with the first term on the R.H.S.~of \eqref{eq:first_decomp_l2}, and averaging over the 
resampling time, one has
\begin{equation*}
\tfrac{1}{N}\,\Exp\,[\,\{\,\overline{\varphi}(\check{X}_{d,1}^1)\,\}^2\,] = 
\tfrac{1}{N}\,\sum_{i=1}^N \Exp\,[\,\overline{w}^{i}_{u(d)}\,\{\,\overline{\varphi}(X_{d,1}^i)\,\}^2\,]
\end{equation*}
 where we have set  $u(d)=l_d(t_{m^*}(d)):d$. Recall that  $\overline{w}_{u(d)}^{i}$ 
denote the normalized weights. By the 
asymptotic independence result in Lemma \ref{lem:growth}(ii)
we have that
\begin{equation*}
\lim_{d\rightarrow\infty} \tfrac{1}{N}\,\sum_{i=1}^N \Exp\,[\,\overline{w}_{u(d)}^{i}\,\{\,\overline{\varphi}(X_{d,1}^i)\,\}^2\,]
=\tfrac{1}{N}\,\Exp\,\Big[\,\sum_{i=1}^N \tfrac{ e^{ Z^i } }{\sum_{l=1}^N e^{Z^l} }\,\{\,\overline{\varphi}(Y^i)\,\}^2\,\Big] = 
\tfrac{\mathbb{V}\textrm{ar}_{\pi}[\varphi]}{N}\ ,
\end{equation*}
where $\{Z^i\}_{i=1}^{N}$ are  i.i.d.~from $N(0,\sigma^2_{t_{m^*-1}:1})$
and, independently, $Y^1,\dots,Y^N$ i.i.d.~from $\pi$.
We now look at the second term on the R.H.S.~of \eqref{eq:first_decomp_l2}. 
Averaging over the resampling index and invoking again the asymptotic independence result of 
Lemma \ref{lem:growth}(ii) 
we have:
\begin{align}
\Exp\,[\,\overline{\varphi}(\check{X}_{d,1}^1)\,&\overline{\varphi}(\check{X}_{d,1}^2)\,] 
 = \sum_{i=1}^{N}\Exp\,[\,\overline{\varphi}^2(X_{d,1}^{i})\,(\overline{w}^{i}_{u(d)})^2\,]+\sum_{i\neq l}
\Exp\,[\,\overline{\varphi}(X_{d,1}^{i})\,\overline{\varphi}(X_{d,1}^{l})\,\overline{w}_{u(d)}^{i}\,
\overline{w}_{u(d)}^{l}\,]\rightarrow 
\nonumber \\
& \pi(\overline{\varphi}^2)\,\Exp\,\Big[\,
 \sum_{i=1}^N \tfrac{ e^{ 2Z^i } }{(\sum_{l=1}^N e^{Z^l})^2 }\,
 \Big] + 0 \equiv  N\,\pi(\overline{\varphi}^2)\,\Exp\,\Big[\,
 \tfrac{ e^{ 2Z^1 } }{(\sum_{l=1}^N e^{Z^l})^2 }\,
 \Big]
\label{eq:abc}
\end{align}
for random variables $\{Z^i\}_{i=1}^{N}$ as defined above (in the last calculation we took advantage of exhangeablity).
We have the decomposition (writing $\sigma^2\equiv\sigma^2_{t_{m^*-1}:1}$ for notational convenience):
\begin{align*}
 \tfrac{ e^{ 2Z^1 } }{(\sum_{l=1}^N e^{Z^l})^2 }
  =    \tfrac{1}{N^2}\,\tfrac{e^{2Z^1}}{e^{\sigma^2}} + e^{-\sigma^2}\,
\tfrac{ e^{ 2Z^1 } }{(\sum_{l=1}^N e^{Z^l})^2 }\,\Big(\, 
e^{\sigma^2} -  \big(\, \tfrac{\sum_{l=1}^{N}e^{Z^l}}{N}\, \big)^2 \,\Big)\ .
\end{align*}
We concentrate on the second term. Using  Holder inequality we have:
\begin{align*}
\Exp\,\Big[\,\Big|\tfrac{ e^{ 2Z^1 } }{(\sum_{l=1}^N e^{Z^l})^2 }\,\Big(\, 
e^{\sigma^2} -  \big(\, \tfrac{\sum_{l=1}^{N}e^{Z^l}}{N}\, \big)^2 \,\Big)\,\Big|\,\Big]
\le 
\Exp^{\frac{2}{3}}\,\big[\,\tfrac{ e^{ 3Z^1 } }{(\sum_{l=1}^N e^{Z^l})^3 }\,\big]\,\,
\Exp^{\frac{1}{3}}\,\Big[\,
\Big|\, 
e^{\sigma^2} -  \big(\, \tfrac{\sum_{l=1}^{N}e^{Z^l}}{N}\, \big)^2 \,\Big|^3
\,\Big]
\end{align*}
Setting  $Z^{(1)}:=\min_{1\le i\le N}Z^{i}$ we get that (using also Cauchy-Schwarz):
\begin{equation*}
\Exp\,\big[\,\tfrac{ e^{ 3Z^1 } }{(\sum_{l=1}^N e^{Z^l})^3 }\,\big]
\le \tfrac{1}{N^3}\,\Exp\,[\,e^{3Z^1-3Z^{(1)}}\,] \le 
\tfrac{1}{N^3}\,\Exp^{\frac{1}{2}}\,[\,e^{6Z^1}\,]\,\Exp^{\frac{1}{2}}\,
[\,e^{-6Z^{(1)}}\,]\ .
\end{equation*}
By standard results on order statistics the pdf of $Z^{(1)}$ is upper bounded 
by $N$ times the pdf of $N(0,\sigma^2)$. So, we have that:
\begin{equation*}
\Exp\,[\,e^{-6Z^{(1)}}\,] \leq N e^{18\sigma^2}\ .
\end{equation*}
By adding and subtracting $e^{\sigma^2}$ in the summand and multiplying the square, one can use Minkowski and the   Marcinkiewicz Zygmund inequality to obtain:
\begin{equation*}
\Exp^{\frac{1}{3}}\,\Big[\,
\Big|\, 
e^{\sigma^2} -  \big(\, \tfrac{\sum_{l=1}^{N}e^{Z^l}}{N}\, \big)^2 \,\Big|^3
\,\Big] \leq \frac{M e^{6\sigma^2}}{N^{1/2}}\ ,
\end{equation*}
for some $M<\infty$ that does not depend upon $N$ or $\sigma^2$.
Putting together the above arguments, we have shown that the right-hand part of the R.H.S.~of \eqref{eq:abc}, when $d\rightarrow\infty$, is upper-bounded by the quantity
$\mathbb{V}\textrm{ar}_{\pi}(\varphi)(\frac{1}{N}e^{\sigma^2}+
Me^{17\sigma^2}\frac{1}{N^{7/6}} )$
which completes the proof.
\end{proof}

\subsection{Normalizing Constants}
\label{app:nc}

\begin{proof}[Proof of Theorem \ref{theo:nc}]
By the expression of the normalized variance (and the fact that the different  particles are i.i.d.), one can re-center to rewrite:
\begin{equation*}
\mathbb{V}_2(\gamma_d(1)) = \mathbb{E}\,\big[\,\big(\tfrac{\overline{\gamma}^N_d(1)}{\overline{\gamma}_d(1)}-1\,\big)^2\,\big]
\end{equation*}
with
\begin{equation*}
\overline{\gamma}^N_d(1) = \tfrac{1}{N}\sum_{i=1}^N e^{\frac{1}{d}\sum_{j=1}^d
G_j^i}\ ; \quad
\overline{\gamma}_d(1)=  \mathbb{E}\,\big[\,e^{\frac{1}{d}\sum_{j=1}^d
G_j^1 }\,\big]\ ,
\end{equation*}
where we have now set
\begin{equation}
\label{eq:gg}
G_j^i = (1-\phi_0)\sum_{n=0}^{d-1}\big(\,g(x_{n,j}^i)-\mathbb{E}\,[\,g(X_{n,j}^i)\,]\,\big)
\end{equation}
and $i\in\{1,\dots,N\}$, $j\in\{1,\dots,d\}$.
We have that:
\begin{align}
\Exp\,\Big[\,\big(\tfrac{\overline{\gamma}^N_d(1)}{\overline{\gamma}_d(1)}-1\big)^2\,\Big]  &=  1 - \tfrac{2}{\overline{\gamma}_d(1)}\,\Exp\,[\,\overline{\gamma}^N_d(1)\,] + \tfrac{1}{\overline{\gamma}_d(1)^2}\,
\Exp\,[\,\overline{\gamma}^N_d(1)^2\,] \nonumber \\
& \equiv -1 + \tfrac{1}{\overline{\gamma}_d(1)^2}\,
\Exp\,[\,\overline{\gamma}^N_d(1)^2\,] \label{eq:eqeq}
\end{align}
where we have used the unbiasedness property (i.e.\@ $\Exp\,[\,\overline{\gamma}^N_d(1)\,] = \overline{\gamma}_d(1)$) of the normalizing constant, see e.g.\@ \cite{delmoral}.
We define
%
$Z^{i}_{d} = \frac{1}{d}\,\sum_{j=1}^{d}G_{j}^{i}$
%
for $G_{j}^{i}$ defined in (\ref{eq:gg}) and $1\le i \le N$.
Thus, due to $Z_{d}^{i}$'s being i.i.d., we have:
\begin{equation*}
\Exp\,[\,\overline{\gamma}^N_d(1)^2\,] = \tfrac{1}{N}\,\Exp\,[\,e^{2Z^{1}_{d}}\,] + (1-\tfrac{1}{N})\,\Exp^2\,[\,e^{Z^{1}_{d}}\,]\ .
\end{equation*}
%
%
%
By Lemma \ref{lem:CLT}(iii), applied when $t_{k-1}(d)\equiv \phi_0$ 
and $t_{k}(d)\equiv 1$, one has that:
\begin{equation*}
\Exp\,[\,e^{2Z^{1}_{d}}\,] \rightarrow \exp\{2\sigma_{\phi_0:1}^2\}\ ;\quad 
\Exp\,[\,e^{Z^{1}_{d}}\,] \rightarrow \exp\{\tfrac{1}{2}\sigma_{\phi_0:1}^2\}\ .
\end{equation*}
Using these limits in (\ref{eq:eqeq}) and recalling also that $\overline{\gamma}_d(1)\equiv \Exp\,[\,e^{Z^{1}_{d}}\,]$,
gives the required result.
\end{proof}
%
%
%
%
%
\begin{proof}[Proof of Theorem \ref{theo:nc1}]
Denote:
\begin{equation}
\overline{\gamma}_{d,k}^N(1) =  \tfrac{1}{N}\sum_{i=1}^Ne^{\frac{1}{d}
\sum_{j=1}^d G_{k,j}^i}\ ; \quad
\overline{\gamma}_{d,k}(1)=  \mathbb{E}_{\pi_{t_{k-1}(d)}^{\otimes Nd}}\big[\,
e^{\frac{1}{d}
\sum_{j=1}^d G_{k,j}^1}\,\big]\ ,
\label{eq:gg2}
\end{equation}
for the standardised $G_{k,j}^{i}$ in (\ref{eq:defineGG}).
We look at the relative $\mathbb{L}_2$-error:
$$
\mathbb{V}_2\Big(\prod_{k=1}^{m^*+1}\gamma_{d,k}(1)\Big) = 
\mathbb{E}\,\Big[\,\bigg(\,\prod_{k=1}^{m^*+1}\tfrac{\overline{\gamma}_{d,k}^N(1)}
{\overline{\gamma}_{d,k}(1)}-1\,\Big)^2\,\Big]\ .
$$
Using the unbiased property of normalising constants,  see e.g.\@ \cite{delmoral},  we have:
\begin{equation*}
\Exp\,\Big[\,
\big(\,\prod_{k=1}^{m^*+1}\tfrac{\overline{\gamma}_{d,k}^N(1)}{\overline{\gamma}_{d,k}(1)}-1\,\big)^2
\,\Big]
 = \Exp\,\Big[
\prod_{k=1}^{m^*+1}\tfrac{\overline{\gamma}_{d,k}^N(1)^2}{\overline{\gamma}_{d,k}(1)^2}\,\Big] - 1\ .
\end{equation*}
For notational convenience, we set:
\begin{equation*}
\Delta_{k,d}: =  \tfrac{\overline{\gamma}_{d,k}^N(1)^2}{\overline{\gamma}_{d,k}(1)^2}\ ;
\quad 
\delta_{k,d} := \Exp_{\pi_{t_{k-1}(d)}^{\otimes N}}\Big[\,\tfrac{\overline{\gamma}_{d,k}^N(1)^2}{\overline{\gamma}_{d,k}(1)^2}\,\Big]\ ;\quad 
\Delta_{1:k,d}= \prod_{q=1}^{k}\Delta_{q,d}\ ;\quad \delta_{1:k,d}= \prod_{q=1}^{k}\delta_{q,d}\ .
\end{equation*}
Following the definitions of $\overline{\gamma}_{d,k}^N(1)$ and $\overline{\gamma}_{d,k}(1)$
in \eqref{eq:gg2}, and exploiting independence among particles under $\pi_{t_{k-1}(d)}^{\otimes
 Nd}$, we have that:
\begin{align*}
 \Exp_{\pi_{t_{k-1}(d)}^{\otimes
 Nd}}\Big[\,\tfrac{\overline{\gamma}_{d,k}^N(1)^2}{\overline{\gamma}_{d,k}(1)^2}\,&\Big]
= \frac{\frac{1}{N}\,\Exp\,[\,e^{\frac{2}{d}\,\sum_{j=1}^{d} G^{1}_{k,j}}\,]+ \big(1-\frac{1}{N} \big)\Exp^2\,[\,e^{\frac{1}{d}\,\sum_{j=1}^{d} G^{1}_{k,j}}\,]}
{\Exp^2\,[\,e^{\frac{1}{d}\,\sum_{j=1}^{d} G^{1}_{k,j}}\,]}\nonumber \\ 
 &\rightarrow e^{-\sigma^{2}_{t_{k-1}:t_{k}}}\,\big[\,e^{2\sigma^{2}_{t_{k-1}:t_{k}}}\,\tfrac{1}{N} +
 (1-\tfrac{1}{N})\,e^{\sigma^{2}_{t_{k-1}:t_{k}}}\,\big] \ ,\label{eq:anteste}
\end{align*}
with the limit obtained from Lemma \ref{lem:CLT}(iii).
Therefore:
\begin{equation*}
 \delta_{1:(m^{*}+1),d} \rightarrow
e^{-\sigma_{\phi_0:1}^2} \prod_{k=1}^{m^{*}+1} \big[\,\tfrac{1}{N}\,e^{2\sigma^2_{t_{k-1}:t_k}}+(1-\tfrac{1}{N})\,e^{\sigma^2_{t_{k-1}:t_k}}\,\big] \ .
\end{equation*}
Thus, it suffices to show that the following difference goes to zero as $d\rightarrow\infty$:
\begin{equation*}
A_d: =\big|\,\Exp\,[\,\Delta_{1:(m^*+1),d}\,]
- \delta_{1:(m^*+1),d}\,\big|\ .
\end{equation*}
Now, note that a simple identity gives that:
\begin{equation*}
A_d = \bigg|\,\sum_{k=1}^{m^{*}+1}\Exp\,\big[\,
\Delta_{1:(k-1),d}
\big(\,\Exp\,[\,\Delta_{k,d}|\mathscr{F}_{t_{k-1}(d)}^N\,] - \delta_{k,d}\,\big)\,\big]\cdot  
\delta_{(k+1):(m^*+1),d}
\,\bigg|\ ,
\end{equation*}
under the conventions that $\Delta_{1:0,d}=\delta_{(m^*+2):(m^*+1)}=1$.
Applying Cauchy-Schwarz yields the following upper-bound:
\begin{align*}
\Exp\,\big[\,\Delta_{1:(k-1),d}\,
\big|\,\Exp\,[\,\Delta_{k,d}&|\mathscr{F}_{t_{k-1}(d)}^N\,] - \delta_{k,d}\,\big|\,\big] \le \\
&\Exp^{1/2}\big[\,
\Delta_{1:(k-1),d}^{2}\,\big]\cdot 
\Exp^{1/2}\,\big[\,
|\,\Exp\,[\,\Delta_{k,d}\mid \mathscr{F}_{t_{k-1}(d)}^N\,] - \delta_{k,d}\,|^{2}\,
\big]\ .
\end{align*}
Via Lemma \ref{lem:cond_conv} the second of the terms in the bottom line vanishes in the limit, so 
it suffices to show that 
 the first term in the bottom line is upper bounded uniformly in~$d$.
Using the Cauchy-Schwarz inequality, we have that:
\begin{equation*}
\Exp\big[\,
\Delta_{1:(k-1),d}^{2}\,\big] \le 
\prod_{q=1}^{k-1}\Exp^{1/2}\,[\,\Delta_{q,d}^{4}\,]\ .
\end{equation*}
Recalling the  definition of $\Delta_{k,d} =  \frac{\overline{\gamma}_{d,k}^N(1)^2}{\overline{\gamma}_{d,k}(1)^2}$ from (\ref{eq:gg2}), 
using triangle inequality for norms we have:
\begin{equation*}
 \Exp\,[\,\overline{\gamma}_{d,q}^N(1)^{8}\,] \le \Big( \tfrac{1}{N} 
\sum_{i=1}^{N} \Exp^{1/8}\,\big[\,
e^{ \frac{8}{d}\sum_{j=1}^{d} G_{q,j}^{i}}
\,\big]\,
\Big)^{8} 
\end{equation*}
%
%
%
Now, we complete via Lemma \ref{lem:ui_exp}.
%
\end{proof}
\begin{proof}[Proof of Proposition \ref{theo:asymp_indep}]
To simplify the notation we drop $i$ for the particle number and define:
\begin{equation*}
\mathcal{G}_{l,d} = \sum_{j=1}^d G_{l,j}\ ,
\end{equation*}
for $1\le l\le k$.
Our proof proceeds by induction. For $k=1$, the result follows by Lemma \ref{lem:CLT}(iii).
Assume that the result holds at time $k-1\ge 1$. Then we have the simple decomposition:
\begin{align}
\Exp\,[\,e^{\sum_{l=1}^{k}c_l\,\mathcal{G}_{l,d}/d} \,] = 
\Exp\,&\Big[\,
\Exp\,[\,e^{c_k\,\mathcal{G}_{k,d}/d}\,|\,\mathscr{F}_{t_{k-1}(d)}^N\,]\,
\big\{\,e^{\sum_{l=1}^{k-1}c_l\,\mathcal{G}_{l,d}/d}-
\Exp\,[\,e^{\sum_{l=1}^{k-1}c_l\,\mathcal{G}_{l,d}/d}\,]\,\big\}\,\Big]\nonumber \\ 
& +\Exp\,[\,e^{\sum_{l=1}^{k-1}c_l\,\mathcal{G}_{l,d}/d}\,]\,\,
\Exp\,[\,e^{c_k\,\mathcal{G}_{k,d}/d}\,]
\label{eq:ind_theorem_master}\ .
\end{align}
We begin by dealing with the first term on the R.H.S.~of \eqref{eq:ind_theorem_master}.
By Lemma \ref{lem:cond_exp_conv} we have that:
\begin{equation*}
\Exp\,[\,e^{c_k\,\mathcal{G}_{k}/d}\,|\,\mathscr{F}_{t_{k-1}(d)}^N\,] - 
\Exp_{\pi_{t_{k-1}(d)}^{\otimes d}}\,[\,e^{c_k \mathcal{G}_{k}/d}\,] \stackrel{\mathbb{P}}{\longrightarrow} 0\ ,
\end{equation*}
whereas from Lemma \ref{lem:CLT}(iii) we have: 
\begin{equation}
\label{eq:here}
\Exp_{\pi_{t_{k-1}(d)}^{\otimes d}}\,[\,e^{c_k \mathcal{G}_{k}/d}\,] \rightarrow e^{\frac{1}{2}\,c_k^2\,\sigma^2_{t_{k-1}:t_k}} \ .
\end{equation}
Moreover, by the induction hypothesis:
\begin{equation*}
 \big\{\,e^{\sum_{l=1}^{k-1}c_l\,\mathcal{G}_{l,d}/d}-
\Exp\,[\,e^{\sum_{l=1}^{k-1}c_l\,\mathcal{G}_{l,d}/d}\,]\,\big\}\, \Rightarrow
 \,e^{\sum_{l=1}^{k-1} c_l X_l} - e^{\frac{1}{2}\sum_{l=1}^{k-1}c_l^2\sigma^2_{t_{l-1}:t_l}}\ .
\end{equation*}
The expression in the expectation of the first term of \eqref{eq:ind_theorem_master} is uniformly integrable: indeed, careful and repeated (but otherwise 
straightforward) use of H\"older and Jensen inequalities will eventually give that:
\begin{align*}
\Big|\,&\Exp\,[\,e^{c_k\,\mathcal{G}_{k,d}/d}\,|\,\mathscr{F}_{t_{k-1}(d)}^N\,]\,
\big\{\,e^{\sum_{l=1}^{k-1}c_l\,\mathcal{G}_{l,d}/d}-
\Exp\,[\,e^{\sum_{l=1}^{k-1}c_l\,\mathcal{G}_{l,d}/d}\,]\,\big\}\,\Big|_{\mathbb{L}_{1+\epsilon}} \\
&\le M\,\prod_{l=1}^{k-1}\big(\,\Exp\,[\,e^{\sum_{l=1}^{k}c'_l\,\mathcal{G}_{l,d}/d}\,]\,
\big)^{1/(1+\delta_l)}
\end{align*}
for positive constants $c'_{1:k}$, $\delta_{1:k}$, $M$ independent of $d$.
 As a consequence, convergence in distribution 
implies also convergence of  expectations:
\begin{align*}
\Exp\,\Big[\,
\Exp\,[\,e^{c_k\,\mathcal{G}_{k,d}/d}\,|\,\mathscr{F}_{t_{k-1}(d)}^N\,]\,
\big\{\,e^{\sum_{l=1}^{k-1}c_l\,\mathcal{G}_{l,d}/d}-
\Exp\,[\,e^{\frac{1}{d}\sum_{l=1}^{k-1}c_l\,\mathcal{G}_{l,d}/d}\,]\,\big\}\,\Big] \\
\rightarrow \Exp\,[\,e^{c_k^2\sigma^2_{t_{k-1}:t_k}/2}\,\{\,e^{\sum_{l=1}^{k-1} c_l X_l} - e^{\frac{1}{2}\sum_{l=1}^{k-1}c_l^2\sigma^2_{t_{l-1}:t_l}}\}\,] \equiv 0\ .
\end{align*}
Now turning to the second term on the R.H.S.\@ of \eqref{eq:ind_theorem_master}, we work as follows:
\begin{align*}
\Exp\,[\,e^{c_k\,\mathcal{G}_{k,d}/d}\,]  = 
\Exp\,\Big[\,
&\Exp\,[\,e^{c_k\,\mathcal{G}_{k,d}/d}\,|\,\mathscr{F}_{t_{k-1}(d)}^N\,]\, - 
\Exp_{\pi_{t_{k-1}(d)}^{\otimes d}}[\,e^{c_k\,\mathcal{G}_{k,d}/d}\,]\,\Big] +
\Exp_{\pi_{t_{k-1}(d)}^{\otimes d}}[\,e^{c_k\,\mathcal{G}_{k,d}/d}\,]\\ 
&\rightarrow 0 + e^{\frac{1}{2}\,c_k^2\,\sigma^2_{t_{k-1}:t_k}}\ ,
\end{align*}
from Lemma \ref{lem:cond_exp_conv} and (\ref{eq:here}).
We can thus deduce by the induction hypothesis that
\begin{equation*}
\Exp\,[\,e^{\sum_{l=1}^{k-1}c_l\,\mathcal{G}_{l,d}/d}\,]\,\,
\Exp\,[\,e^{c_k\,\mathcal{G}_{k}/d}\,] \rightarrow e^{\frac{1}{2}\sum_{l=1}^{k}c_l^2\sigma^2_{t_{l-1}:t_l}}
\equiv \prod_{l=1}^k \Exp\,[\,e^{c_l Z^l}\,] 
\end{equation*}
which completes the proof.
\end{proof}

\begin{lem}
\label{lem:cond_conv}
Assume (A\ref{hyp:A}-\ref{hyp:B}) and $g\in\mathcal{B}_b(E)$. Then 
for any $\epsilon>0$, $N\geq 1$ and $1\leq k \leq m^*+1$:
\begin{equation*}
\Exp\,\big[\,\tfrac{\overline{\gamma}_{d,k}^N(1)^2}
{\overline{\gamma}_{d,k}(1)^2}\,\big|\,\mathscr{F}_{t_{k-1}(d)}^N\,\big]
-\Exp_{\pi_{t_{k-1}(d)}^{\otimes Nd}}\,\big[\,\tfrac{\overline{\gamma}_{d,k}^N(1)^2}{\overline{\gamma}_{d,k}(1)^2}\,\big]
\rightarrow 0 \ ,\quad   \textrm{in }\,\, \mathbb{L}_{1+\epsilon}\ .
\end{equation*}
\end{lem}

\begin{proof}
Due to conditional independence among particles given $\mathscr{F}_{t_{k-1}(d)}^N$, we have:
\begin{align}
&\Exp\,\big[\,\overline{\gamma}_{d,k}^N(1)^2\,\big|\,\mathscr{F}_{t_{k-1}(d)}^N\,\big] = \label{eq:aga}
\\ &= \tfrac{1}{N^2}\,\Big( \,\Exp\,\big[\,
\sum_{i=1}^N e^{\frac{2}{d}\sum_{j=1}^d G_{k,j}^i}|\,\mathscr{F}_{t_{k-1}(d)}^N\,\big] + 
\sum_{i\neq m}\Exp\,[\, e^{\frac{1}{d}\sum_{j=1}^d G_{k,j}^i} |\,\mathscr{F}_{t_{k-1}(d)}^N\,\big]
\,\Exp\,\big[\,e^{\frac{1}{d}\sum_{j=1}^d G_{k,j}^{m}}\,
\big|\,\mathscr{F}_{t_{k-1}(d)}^N\,\big]\,\Big) \ .\nonumber
\end{align}
%
%
Now, for any constant $c\ge 1$ we have $\sup_d\,\Exp_{\pi^{\otimes Nd}_{t_{k-1}(d)}}\big[\,e^{\frac{c}{d}\sum_{j=1}^d G_{k,j}^{i}}\,\big]<\infty$
from Lemma \ref{lem:ui_exp}, so it suffices to prove that for any constant $c\ge 1$, as $d\rightarrow\infty$:
\begin{equation}
\Exp\,\big[\,e^{\frac{c}{d}\sum_{j=1}^d G_{k,j}^{i}}\big|\mathscr{F}_{t_{k-1}(d)}^N\,\big] - \Exp_{\pi_{t_{k-1}(d)}^{\otimes Nd}}\big[\,e^{\frac{c}{d}\sum_{j=1}^d G_{k,j}^{i}}\,\big]
\rightarrow 0\ , \quad \textrm{in }\,\, \mathbb{L}_{2(1+\epsilon)}\ .
        \end{equation}
The factor of two in the norm arises as own has to use Cauchy-Schwarz to separate the product terms 
on the R.H.S.\@ of (\ref{eq:aga}). 
Now, Lemma \ref{lem:cond_exp_conv}  
established the above convergence in probability; this together with uniform integrability implied by Lemma \ref{lem:ui_exp} establishes the result. 
%
\end{proof}

\begin{lem}\label{lem:cond_exp_conv}
Assume (A\ref{hyp:A}-\ref{hyp:B}) and that $g\in\mathcal{B}_b(E)$. Then, we have that 
for any $N\geq 1$, $i\in\{1,\dots,N\}$,  $k\in\{1,\dots,m^*+1\}$ and $c\in\mathbb{R}$:
\begin{equation*}
\Exp\,\big[\,e^{\frac{c}{d}\sum_{j=1}^d G_{k,j}^i}\,\big|\,\mathscr{F}_{t_{k-1}(d)}^N\,\big]
-\Exp_{\pi_{t_{k-1}(d)}^{\otimes Nd}}\big[\,e^{\frac{c}{d}\sum_{j=1}^d G_{k,j}^i}\,\big]
\stackrel{\mathbb{P}}{\longrightarrow} 0\  .
\end{equation*}
\end{lem}
\begin{proof}
By the conditional independence along $j$, we have:
\begin{equation*}
\Exp\,\big[\,e^{\frac{c}{d}\sum_{j=1}^d G_{k,j}^{i}}\,\big|\,\mathscr{F}_{t_{k-1}(d)}^N\,\big]
= \prod_{j=1}^d \Exp_{\check{X}_{l_d(t_{k-1}(d)),j}}\big[\,e^{\frac{c}{d}\,G_{k,j}^{i}}\,]\ .
\end{equation*}
We will now omit various subscripts/superscripts to simplify the notation, using also
$\Exp_{\pi}\equiv \Exp_{\pi_{t_{k-1}(d)}}$ and $\Exp_{\check{X}_{0,j}}\equiv \Exp_{\check{X}_{l_d(t_{k-1}(d)),j}}$.
We can rewrite:
\begin{align}
&\Exp\,\big[\,e^{\frac{c}{d}\sum_{j=1}^d G_{j}}\,\big|\,\mathscr{F}^N\,\big]
-\Exp_{\pi^{\otimes Nd}}\big[\,e^{\frac{c}{d}\sum_{j=1}^d G_{j}}\,\big]  = \nonumber \\
&= \Big(\,\prod_{j=1}^d \Exp_{\pi}\,[\,e^{c\,G_{j}/d}\,]\,\Big)\,\Big[\,\prod_{j=1}^d\Big\{\,\tfrac{\{\,\Exp_{\check{X}_{0,j}}
-\Exp_{\pi}\,\}[\,e^{c\,G_{j}/d}\,]\,}{\Exp_{\pi}[\,e^{c\,G_{j}/d}\,]} + 1\,\Big\} - 1\,\Big]
\label{eq:main_diff_lem_exp}\ .
\end{align}
From Lemma \ref{lem:CLT}(iii) it follows that 
$
\prod_{j=1}^d \Exp_{\pi}\,[\,e^{c\,G_{j}/d}\,]\rightarrow e^{\frac{1}{2}\,c^2\,\sigma^{2}_{t_{k-1}:t_k}}$,
hence we can now concentrate on the second factor-term on the R.H.S.\@ of \eqref{eq:main_diff_lem_exp}.
We will replace the product with a sum using logarithms. To that end define:
\begin{equation*}
\beta_j(d) := 
\tfrac{\{\,\Exp_{\check{X}_{0,j}}
-\Exp_{\pi}\,\}[\,e^{c\,G_{j}/d}\,]\,}{\Exp_{\pi}[\,e^{c\,G_{j}/d}\,]}\ .
\end{equation*}
%
Note that since $g\in\mathcal{B}_b(E)$, 
we have that $G_j/d$ is bounded from above and below, so there exist an $\epsilon>0$ and $M>0$ such that:
\begin{equation}
\label{eq:bbb}
-1+\epsilon \leq \beta_j(d) \leq M < \infty\ .
\end{equation}
We need to prove that $e^{\sum_{j=1}^d\log(1+\beta_j(d))} - 
1\stackrel{\mathbb{P}}{\rightarrow} 0$.
We consider a second order Taylor expansion of the exponent:
\begin{equation}
\label{eq:156}
\sum_{j=1}^d\log(1+\beta_j(d)) = \sum_{j=1}^d\big\{\,\beta_j(d) - \tfrac{1}{2}\,\tfrac{1}{(1+\xi_j(d))^2}\,\beta_j^2(d)\,\big\}
\end{equation}
where $\xi_j(d)\in[\,0\wedge \beta_j(d) ,0\vee \beta_j(d)\,]$.
By Lemma \ref{lem:betaj_conv} we have that:
\begin{equation*}
\sum_{j=1}^d \beta_j(d) \stackrel{\mathbb{P}}{\rightarrow} 0 \ ; 
\quad 
\sum_{j=1}^d \beta^{2}_j(d) \stackrel{\mathbb{P}}{\rightarrow} 0\  .
\end{equation*}
Since $\xi_j(d)$'s are bounded due to (\ref{eq:bbb}), these two results imply via the Taylor expansion in (\ref{eq:156})
that also
$
 \sum_{j=1}^d\log(1+\beta_j(d)) \stackrel{\mathbb{P}}{\rightarrow} 0$.
Due to the continuity of the exponential function, this implies now that 
$e^{\sum_{j=1}^d\log(1+\beta_j(d))} - 1\Rightarrow  0 $
and the proof is now complete since weak convergence to a constant implies convergence in probability.
\end{proof}

\begin{lem}\label{lem:betaj_conv}
Assume (A\ref{hyp:A}-\ref{hyp:B}) and $g\in\mathcal{B}_b(E)$. Then we have that for any $N\geq 1$, $i\in\{1,\dots,N\}$, $k\in\{1,\dots,m^*+1\}$ and $c\in\mathbb{R}$:
\begin{align*}
&(i)\quad \quad \sum_{j=1}^{d} \frac{\{\,\Exp_{\check{X}_{l_d(t_{k-1}(d)),j}^i}-\,\Exp_{\pi_{t_{k-1}(d)}}\,\}[\,e^{\,c\,G_{k,j}^i/d}\,]}
{\Exp_{\pi_{t_{k-1}(d)}}[\,e^{c\,G_{k,j}^i/d}\,]}
\rightarrow 0 \ ,\quad \textrm{in }\,\,\mathbb{L}_1\ . \\
&(ii)\quad \quad \sum_{j=1}^{d} \bigg(\, \frac{\{\,\Exp_{\check{X}_{l_d(t_{k-1}(d)),j}^i}-\,\Exp_{\pi_{t_{k-1}(d)}}\,\}[\,e^{\,c\,G_{k,j}^i/d}\,]}
{\Exp_{\pi_{t_{k-1}(d)}}[\,e^{c\,G_{k,j}^i/d}\,]}\,\bigg)^2
\rightarrow 0 \ ,\quad \textrm{in }\,\,\mathbb{L}_1\ .
\end{align*}
\end{lem}

\begin{proof}
To simplify the presentation, we drop many super/subscripts: that is, we write the quantity of interest as:
\begin{equation*}
\frac{\{\,\Exp_{\check{X}_{0,j}}-\Exp_{\pi}\,\}[\,e^{c\,G_{j}/d}\,]}
{\Exp_{\pi}\,[\,e^{c\,G_{j}/d}\,]} \ .
\end{equation*}
Note that
$\Exp_{\pi}\,[\,e^{c\,G_{j}/d}\,]\equiv \Exp_{\pi}\,[\,e^{c\,G_{1}/d}\,]$. 
Since $|\,g\,|$ is bounded, $|\,G_{1}/d\,|$ is also bounded, so $\Exp_{\pi}\,[\,e^{c\,G_{1}/d}\,]$
is lower and upper bounded by positive constants and can be ignored in the calculations.
We will be using the second-order Taylor expansion:
\begin{equation}
\label{eq:TTaylor}
e^{c\,G_j/d} = 1 + \tfrac{c\,G_{j}}{d} + \tfrac{1}{2}\,e^{\xi_j(d)}\,\big(\,\tfrac{c\,G_{j}}{d}\,\big)^2 \ , 
\end{equation}
where $\xi_j(d)\in[\,0\wedge \tfrac{c\,G_{j}}{d}\,, 0\vee \tfrac{c\,G_{j}}{d}\,]$. \vspace{0.2cm}\\
\noindent \emph{Proof of (i)}:\\
The $\mathbb{L}_1$-norm of the variable of interest is upper bounded by (recalling that $\Exp_{\pi}\,[\,G_j\,]
\equiv 0$):
\begin{equation*}
\Exp\,\Big|\sum_{j=1}^d \Exp_{\check{X}_{0,j}}
\big[\,\tfrac{c\,G_{j}}{d}\,\big]\,\Big| + \tfrac{c^2}{2}\,
\Exp\,\Big|\sum_{j=1}^d\big\{\,\Exp_{\check{X}_{0,j}}-\Exp_{\pi}\,\big\}\,\big[\,
e^{\xi_j(d)}\,\big(\,\tfrac{G_{j}}{d}\,\big)^2\,\big]\,\Big|\ .
\end{equation*}
The first term in this bound goes to zero by Lemma \ref{lem:CLT}(iv). Thus considering the second term, we have the trivial inequality (for convenience we set $\sigma^2\equiv \sigma^2_{t_{n-1}:t_n}$):
\begin{align}
&\Exp\,\Big|\sum_{j=1}^d\big\{\,\Exp_{\check{X}_{0,j}}-\Exp_{\pi}\,\big\}
\big[\,e^{\xi_j(d)}\big(\,\tfrac{G_{j}}{d}\,\big)^2\,\big]\Big| \leq 
\label{eq:master_eq} \\ 
&\Exp\,\Big|\sum_{j=1}^d \Exp_{\check{X}_{0,j}}\big[\,\,(e^{\xi_j(d)}-1)\big(\,\tfrac{G_{j}}{d}\,\big)^2\,\big]\,
\Big|
 + \Exp\,\Big|\sum_{j=1}^d \Exp_{\check{X}_{0,j}}\big[\,\big(\,\tfrac{G_{j}}{d}\,\big)^2\,\big] -\sigma^2\,\Big| + \big|\,\sigma^2- \tfrac{1}{d}\,\Exp_{\pi}\,[\,e^{\xi_1(d)}\,G_{1}^2\,]\,\big|
\ . \nonumber
\end{align}
Note that 
\begin{itemize}
\item{$|\xi_1(d)|<M$ (due to the boundedness assumption on $g$)}\ ;
\item{$\xi_1(d)\rightarrow 0$ in distribution (so also in $\mathbb{L}_{p}$ for any $p\ge 1$ due to
the above uniform bound)} \ ;
\item{$\Exp_{\pi}\,[\,G_1^2/d\,]\rightarrow \sigma^2$}\ ,
\end{itemize}
with the last two results following from 
Lemma \ref{lem:CLT}(i,ii). These results, together, imply that 
the last term on the R.H.S.\@ of \eqref{eq:master_eq} goes to zero. 
For the first term on the R.H.S.\@ of \eqref{eq:master_eq} we work as follows.
Since for each $j$, $|G_j/d|$ is bounded  
we have that  $|e^{\xi_j(d)}-1|\leq M\,|\xi_j(d)|\leq M\,|\frac{G_j}{d}|$. 
As a result, using the triangular inequality and then this latter bound we have that:
\begin{align*}
\Exp\,\Big|\sum_{j=1}^d \Exp_{\check{X}_{0,j}}\big[\,(e^{\xi_j(d)}&-1)\big(\,\tfrac{G_{j}}{d}\,\big)^2\,\big]\,\Big|
\le \tfrac{M}{d^3}\sum_{j=1}^d\Exp\,\big[\,\Exp_{\check{X}_{0,j}}|G_j|^{3}\,\big] 
= \tfrac{M}{d^3}\sum_{j=1}^d \Exp\,|G_j|^{3}.
\end{align*}
From Lemma \ref{lem:CLT}(ii) we have that this latter term is upper-bounded by
$\tfrac{M}{d^3}d\,d^{3/2}\rightarrow 0$.
Now, for the second term on the R.H.S.\@ of \eqref{eq:master_eq} 
we work as follows. We have that: 
\begin{align*}
\Exp_{\check{X}_{0,j}}[\,G_j^{2}\,] = 
\Exp_{\check{X}_{0,j}}\big[\,\big(\,G_j-\Exp_{\check{X}_{0,j}}\,[\,G_j\,]\,\big)^2\,\big] + 
\Exp^{2}_{\check{X}_{0,j}}\,[\,G_{j}\,]\ .
\end{align*}
Lemma \ref{lem:CLT}(ii) gives that $|\,\Exp_{\check{X}_{0,j}}[\,G_{j}\,]\,|\le M$, 
so we have that $\tfrac{1}{d^2}\,\sum_{j=1}^{d}\,\Exp^{2}_{\check{X}_{0,j}}\,[\,G_{j}\,]\rightarrow 0$ in $\mathbb{L}_1$. The result now follows from Lemma \ref{lem:CLT}(iv).
\vspace{0.2cm}\\
\noindent \emph{Proof of (ii)}:\\
We will use again the Taylor expansion (\ref{eq:TTaylor}). Clearly, 
the $\mathbb{L}_1$-norm of the random variable of interest is bounded by: 
\begin{equation*}
2\sum_{j=1}^d \Exp\,\Big[\,\big(\,\Exp_{\check{X}_{0,j}}
\,\big[\,\tfrac{G_j}{d}\,\big]\,\big)^2\,\Big]
+
2\sum_{j=1}^d \Exp\,\Big[\,\Big(\,\big\{\,\Exp_{\check{X}_{0,j}}-\Exp_{\pi}\,\big\}\,
\big[\,\tfrac{1}{2}\big(\,\tfrac{G_j}{d}\,\big)^2e^{\xi_j(d)}\,\big]\,\Big)^2\,\Big]\  .
\end{equation*}
The first term goes to zero from the first result in Lemma \ref{lem:CLT}(ii) and the second from 
the second result in Lemma \ref{lem:CLT}(ii) applied here for $p=4$.
\end{proof}

\begin{lem}\label{lem:ui_exp}
Assume (A\ref{hyp:A}-\ref{hyp:B}) and $g\in\mathcal{B}_b(E)$. Then 
we have that for any $N\geq 1$, $i\in\{1,\dots,N\}$ and $k\in\{1,\dots,m^*+1\}$
and any fixed $c\in\mathbb{R}$:
\begin{equation*}
\sup_d\, 
\Exp\,\big[\,e^{\frac{c}{d}\sum_{j=1}^d G_{k,j}^i}\,\big]< \infty\ .
\end{equation*}
\end{lem}

\begin{proof}
To simplify the notation we rewrite the quantity of interest as
\begin{equation*}
\Exp\,\big[\,e^{\frac{c}{d}\sum_{j=1}^d G_{j}}\,\big] \equiv \Exp\,\Big[
\prod_{j=1}^{d}\Exp_{\check{X}_{0,j}}\,\big[\,e^{\frac{c}{d} G_{j}}\,\big]\ .
\,\Big]
\end{equation*}
Applying a second order Taylor expansion for $e^{\frac{c}{d} G_{j}}$ yields that the above is equal to: 
\begin{equation*}
\Exp\,\Big[\,\prod_{j=1}^d\big(\,1+c\,\Exp_{\check{X}_{0,j}}\big[\,\tfrac{G_{j}}{d}\,\big] + \tfrac{c^2}{2}\,\Exp_{\check{X}_{0,j}}\big[\,\big(\,e^{\xi_j(d)}\tfrac{G_{j}}{d}\,\big)^2\,\big]\,\big)
\,\Big]
\end{equation*}
with $\xi_j(d)\in[\,0\wedge \tfrac{c\,G_j}{d},0\vee \tfrac{c\,G_j}{d}\,]$.
Using the fact that $|G_{j}/d|$ is upper bounded by a constant, 
from Lemma \ref{lem:CLT}(ii) we have:
\begin{equation*}
\big|\,c\,\Exp_{\check{X}_{0,j}}\,\big[\,\tfrac{G_{j}}{d}\,\big]\,\big|\leq |c|\,\tfrac{M}{d}\ ; 
\quad \tfrac{c^2}{2}\,\Exp_{\check{X}_{0,j}}\big[\,\big(\,e^{\xi_j(d)}\tfrac{G_{j}}{d}\,\big)^2\,\big]
\leq c^2\,\tfrac{M}{d} \ .
\end{equation*}
Hence, we have that:
$$
\Exp\,\big[\,e^{\frac{c}{d}\sum_{j=1}^d G_{j}}\,\big] \leq \big(1+\tfrac{M}{d}\big)^d
$$
with the latter upper bound converging by standard results in analysis.
\end{proof}

\subsection{Propagation of Chaos}\label{app:A}

\begin{proof}[Proof of Proposition \ref{theo:prop_chaos}]

For simplicity, consider the first $q$ of $N$ particles and $j=1$. 
Then, for a function $F:E^q\rightarrow[0,1]$ we have, using the notation $X_{s(d),1}^{1:q} = (X_{s(d),1}^1,\dots,X_{s(d),1}^q)$:
%
\begin{align}
|\,\Exp\,[\,F(X_{s(d),1}^{1:q})\,]
&-\pi_{s}^{\otimes q}(F)\,| 
 \leq |\,\Exp\,[\,F(X_{s(d),1}^{1:q})\,]-
\Exp_{\pi^{\otimes N}_{t_{k-1}(d)}}[\,F(X_{s(d),1}^{1:q})\,]\,|
\nonumber
\\ & + |\,\Exp_{\pi^{\otimes N}_{t_{k-1}(d)}}[F(X_{s(d),1}^{1:q})]-\pi_{s(d)}^{\otimes q}(F)\,|
 + |\,\pi_{s(d)}^{\otimes q}(F)-\pi_{s}^{\otimes q}(F)\,|\ .
\label{eq:decomp}
\end{align}
The last term on the R.H.S.~goes to zero via the bounded convergence theorem (this follows directly from 
having assumed that $g$ is upper bounded), so we consider the first two terms.
For the first term on the R.H.S.\@ of (\ref{eq:decomp}) one can use conditional expectations and write it as:
\begin{equation*}
\Exp\,\Big[\,\Exp\,[\,F(X_{s(d),1}^{1:q})\mid\mathscr{F}_{t_{k-1}(d)}^N\,] - 
\Exp_{\pi^{\otimes N}_{t_{k-1}(d)}}[\,F(X_{s(d),1}^{1:q})\,]\,\Big] 
\end{equation*}
where $\mathscr{F}_{t_{k-1}(d)}^N$ is the filtration generated by the particle system up to (and including) the $(n-1)^{th}$ resampling time.
The quantity inside the expectation can be equivalently written as:
%
\begin{equation}
\label{eq:dd}
\big\{ k_{u(d)}^{\otimes q}(\check{X}_{l_d(t_{k-1}(d)),1}^{1:q},\,\cdot\,) - 
\big(\pi_{t_{k-1}(d)}k_{u(d)}\big)^{\otimes q}\big\}(F) 
\end{equation}
where we set $u(d)=(l_d(t_{k-1}(d))+1):l_d(s(d))$. 
For $1\le l\le q$ we define the probability measures:
\begin{equation*}
\mu_l = \mu_{l}(d{y_{1:(l-1)}},dy_{(l+1):q}) = \big(\pi_{t_{k-1}(d)}k_{u(d)}\big)^{\otimes (l-1)}\otimes
k_{u(d)}^{\otimes (q-l)}(\check{X}_{l_d(t_{k-1}(d)),1}^{(l+1):q},\,\cdot\,) \ .
\end{equation*}
Notice the simple identity (since intermediate terms in the sum below will cancel out):
\begin{align}
\big\{k_{u(d)}^{\otimes q}&(\check{X}_{l_d(t_{k-1}(d)),1}^{1:q},\,\cdot\,) - 
\big(\pi_{t_{k-1}(d)}k_{u(d)}\big)^{\otimes q}\big\}(dy_{1:q})  \label{eq:identity} \\ &= \sum_{l=1}^{q}
\big( k_{u(d)}(\check{X}_{l_d(t_{k-1}(d)),1}^{l},\cdot ) - 
\pi_{t_{k-1}(d)}k_{u(d)}\big)(dy_l)\otimes \mu_l(d{y_{1:(l-1)}},dy_{(l+1):q})\ .
\nonumber
\end{align}
Since $|F|\le 1$, we have  
 $|\int \mu_l(d{y_{1:(l-1)}},dy_{(l+1):q})F(y_{1:q})|\le 1$ for any $y_l$. 
%
Given this property, 
using the identity (\ref{eq:identity})  we have that the expression in (\ref{eq:dd}) is bounded
in absolute value by:
%
%
\begin{align*}
\bigl|\,\sum_{l=1}^q \int_{\mathbb{R}}\{k_{u(d)}(\check{X}_{l_d(t_{k-1}(d)),1}^{l},\cdot\,) &-\pi_{t_{k-1}(d)}k_{u(d)}\}(dy_l)
\bigg\{ 
\frac{ \int \mu_l(d{y_{1:(l-1)}},dy_{(l+1):q})F(y_{1:q})}
{\sup_{y_l\in \mathbb{R}}|\int \mu_l(d{y_{1:(l-1)}},dy_{(l+1):q})F(y_{1:q})|
}
\bigg\}\,\bigr| \\ 
&\le
 \sum_{l=1}^q\|k_{u(d)}(\check{X}_{l_d(t_{k-1}(d)),1}^{l})-\pi_{t_{k-1}(d)}k_{u(d)}\|_{tv}\ .
\end{align*}
The above total variation bound converges to zero in $\mathbb{L}_1$ as $d\rightarrow\infty$ by Lemma 
\ref{lem:growth}(i), so also the first term on  the R.H.S.\@ of \eqref{eq:decomp} goes 
to zero as $d\rightarrow\infty$.
The second term on  the R.H.S.\@ of~\eqref{eq:decomp} can be treated in a similar manner. One has again the 
identity:
\begin{align*}
\Exp_{\pi^{\otimes N}_{t_{k-1}(d)}}\,&[\,       F(X_{s(d),1}^{1:q})\,]-\pi_{s(d)}^{\otimes q}(F) = \\
& =  \sum_{l=1}^q \int_{\mathbb{R}}\{\,\pi_{t_{k-1}(d)}k_{u(d)}-\pi_{s(d)}\,\}(dx_l)\times
\Big\{ \pi_{s(d)}^{\otimes (l-1)}\otimes
\big(\pi_{t_{k-1}(d)}k_{u(d)}\big)^{\otimes (q-l)} (F(x_l))
\Big\} \\
& \le q\,\|\pi_{t_{k-1}(d)}k_{u(d)}-\pi_{s(d)}\|_{tv}\ .
\end{align*}
This last bound 
which will go to zero by Lemma 
\ref{lem:growth}(i). Hence we conclude.
\end{proof}

\end{document}